\documentclass[journal]{IEEEtran}
\usepackage{amsmath,amsfonts}
\usepackage{xcolor}
\usepackage{amssymb}
\usepackage{algorithmic}
\usepackage{algorithm}
\usepackage{array}
\usepackage[caption=false,font=footnotesize,labelfont=rm,textfont=rm]{subfig}
\usepackage{textcomp}
\usepackage{stfloats}
\usepackage{url}
\usepackage{verbatim}
\usepackage{graphicx}
\usepackage{cite}
\usepackage{amsthm}
\theoremstyle{definition}  
\newtheorem{remark}{Remark}
\newtheorem{lemma}{Lemma}
\newtheorem{theorem}{Theorem}  
\begin{document}

\title{Joint Node Selection and Resource Allocation Optimization for Cooperative Sensing with a Shared Wireless Backhaul}
\author{Mingxin Chen, Ming-Min Zhao, \IEEEmembership{Senior Member,~IEEE,} An Liu, \IEEEmembership{Senior Member,~IEEE,} Min Li,  \IEEEmembership{Member,~IEEE,} and Qingjiang Shi, \IEEEmembership{Member,~IEEE}
\thanks{M. Chen, M. M. Zhao, A. Liu and M. Li are with College of Information Science and Electronic Engineering, Zhejiang University, and also with Zhejiang Provincial Key Laboratory of Information Processing, Communication and Networking, Hangzhou 310027, China (Emails: \{22331108, zmmblack, anliu, min.li\}@zju.edu.cn). Q. Shi is with the School of Software Engineering, Tongji University, Shanghai 201804, China, and also with the Shenzhen Research Institute of Big Data, Shenzhen 518172, China (e-mail: shiqj@tongji.edu.cn). (\emph{Corresponding author: Ming-Min Zhao; Min Li}.)    \par
This work was supported in part by the National Natural Science Foundation of China under Grant 62071416, and in part by the National Natural Science Foundation of China under Grant 62271440.  \par
	
Part of this paper has been accepted at IEEE PIMRC 2024 \cite{mx}, and the additional contributions include the theoretical proof of the convergence of the MCSCA algorithm, computational complexity analysis of the proposed algorithms, along with additional simulation results.}}


\maketitle

\begin{abstract}
In this paper, we consider a cooperative sensing framework in the context of future multi-functional network with both communication and sensing ability, where one base station (BS) serves as a sensing transmitter and several nearby BSs serve as sensing receivers. Each receiver receives the sensing signal reflected by the target and communicates with the fusion center (FC) through a wireless multiple access channel (MAC) for cooperative target localization. To improve the localization performance, we present a hybrid information-signal domain cooperative sensing (HISDCS) design, where each sensing receiver transmits both the estimated time delay/effective reflecting coefficient and the received sensing signal sampled around the estimated time delay to the FC. Then, we propose to minimize the number of channel uses by utilizing an efficient Karhunen-Loéve transformation (KLT) encoding scheme for signal quantization and proper node selection, under the Cramér-Rao lower bound (CRLB) constraint and the capacity limits of MAC. A novel matrix-inequality constrained successive convex approximation (MCSCA) algorithm is proposed to optimize the wireless backhaul resource allocation, together with a greedy strategy for node selection. Despite the high non-convexness of the considered problem, we prove that the proposed MCSCA algorithm is able to converge to the set of Karush-Kuhn-Tucker (KKT) solutions of a relaxed problem obtained by relaxing the discrete variables. Besides, a low-complexity quantization bit reallocation algorithm is designed, which does not perform explicit node selection, and is able to harvest most of the performance gain brought by HISDCS. Finally, numerical simulations are presented to show that the proposed HISDCS design is able to significantly outperform the baseline schemes. 
\end{abstract}


\section{Introduction}
Multi-functional network, which can provide both reliable communication and high-accuracy sensing services, is expected to play a crucial role in many application scenarios of the future 6G system such as autonomous driving, extended reality (XR) and multi-base radar sensing \cite{cui2021integrating,liu2020joint,zhang2021enabling,liu2022survey}. In these scenarios, cooperative sensing by fusing perception data from multiple nodes is considered to be a powerful technique in various applications due to its inherent advantages \cite{liu2022integrated,patwari2005locating,meng2013optimality}. On the one hand, cooperative sensing improves the accuracy of detection by combining information from multiple nodes, as the errors from one node can be mitigated by exploiting the sensing information from other nodes. On the other hand, cooperative sensing is envisioned to reduce the communication overhead significantly while ensure rapid adjustments to node deployment and localization strategy \cite{meng2013optimality}. By sharing information between nodes, the overall cost of sensing networks can be reduced while maintaining or even improving the localization performance \cite{zhu2017localisation}. 

Generally, existing cooperative sensing schemes can be broadly categorized into two types, i.e., information-domain cooperative sensing (IDCS) and signal-domain cooperative sensing (SDCS). In IDCS, each sensing receiver extracts specific location information from the echo signals and then the fusion center (FC) fuses the extracted information  for cooperative target localization  \cite{sayed2005network,xiong2017cooperative}. The common approaches of IDCS are usually based on the time-of-arrival (TOA) \cite{pak2016distributed,wu2015performance}, angle-of-arrival (AOA) \cite{shao2014efficient}, time-difference-of-arrival (TDOA) \cite{wang2011importance}, received signal strength (RSS) \cite{li2006rss}, etc.  Besides, hybrid IDCS methods which combine TOA, TDOA, AOA, or RSS measurements to improve accuracy and robustness, have gained significant research interests \cite{li2018performance,tomic2018robust}. For example, the work \cite{li2018performance} proposed to combine the TOA and AOA information such that both time and angular information measurements are leveraged,  which yields a more accurate estimate in complex environments.
As for the localization performance of IDCS, it is strongly relative to the employed data fusion algorithms \cite{fan2019extended}, which can generally be categorized into two types, i.e., centralized algorithm and distributed algorithm. Centralized algorithms can offer more accurate position estimates in small networks with acceptable complexity \cite{win2011network}, while distributed algorithms are also attractive since they can improve the network robustness and scalability \cite{mazuelas2011information}.
However, the localization accuracy of IDCS is easily affected by non-line-of-sight (NLOS)  propagation environments \cite{aghaie2016localization}. Besides, according to the law of data processing, extracting location information from echo signals always incurs certain information loss, which may lead to localization performance degradation. 

In order to enhance the localization performance, the SDCS scheme has attracted great research interests recently \cite{wang2018parameter,khalili2014cloud,xi2020joint,zhang2023direct}, where each sensing receiver directly sends the echo signals to the FC for cooperative sensing. 
Specifically, in \cite{wang2018parameter}, the authors adopted a uniform quantizer for the echo signals in a cloud multiple-input multiple-output (MIMO) radar system, and Gaussian noise approximation for the quantization error was employed to evaluate the impact of finite backhaul capacity on target localization performance.
The work \cite{khalili2014cloud} considered a multistatic radar setup, where distributed receive antennas were connected to the FC via limited backhaul links, and the localization performance was enhanced by jointly optimizing the code vector and the statistical properties of the noise introduced through backhaul quantization.
In \cite{xi2020joint}, a cooperative radar sensing system based on one-bit sampling of the received radar echoes was studied, and it was shown that the application of one-bit sampling significantly reduces the hardware cost, energy consumption and systematic complexity.
Besides, the work \cite{zhang2023direct}, studied a low-bit direct localization method, where a Cramér-Rao lower bound (CRLB)-based objective function was designed to obtain the optimum quantization thresholds for each receiver, and it was shown that superior localization performance over the IDCS scheme can be achieved.
As a brief summary, we can see that the SDCS scheme is able to achieve better  sensing performance as compared to the IDCS scheme due to the ability of fully utilizing the information contained in the echoes. However, it comes with higher communication overhead and power consumption load, especially when the number of sensing receivers is large. 

Motivated by the above,  a hybrid information-signal domain cooperative sensing (HISDCS) framework is proposed in this paper in pursuit of an efficient tradeoff between  IDCS and  SDCS. In the proposed framework, carefully-quantized echo signals and some extracted location-related  information from different receivers are fused in the FC via a limited-backhaul wireless multiple access channel (MAC). It is worth noting that  by employing a subset of reliable nodes for localization, not only can the interference from unreliable nodes be reduced, but communication overhead can also be saved. Besides, we employ the CRLB  as the sensing performance metric and formulate an optimization problem to minimize the number of channel uses by optimizing the backhaul resource (i.e., quantization bits) allocation and cooperative node selection, under MAC capacity constraints. The main contributions of this work are summarized as follows:

\begin{enumerate}
\item{In the proposed HISDCS scheme, besides the estimated time delay and reflecting coefficient, we propose that  each sensing receiver also transmits the echo signal sampled around the estimated delay with a proper quantization strategy to the FC. A Karhunen-Loéve transformation (KLT) based encoding scheme  is proposed for efficient  echo signal quantization.}
\item{A novel matrix-inequality constrained successive convex approximation (MCSCA) algorithm is proposed for quantization bits allocation and a greedy strategy is designed for node selection. Besides, we  theoretically prove the convergence of the MCSCA algorithm given a feasible initial solution. Additionally, a low-complexity quantization bit reallocation algorithm is proposed, which retains most of the performance gain offered by HISDCS.}
\item{Simulations results are presented to show the superiority of the proposed HISDCS scheme over some baseline designs, and it is shown that the proposed HISDCS scheme is able to achieve an efficient tradeoff between localization performance and communication cost. }
\end{enumerate}

\begin{figure}[htbp]
\centerline{\includegraphics[width=0.85\linewidth]{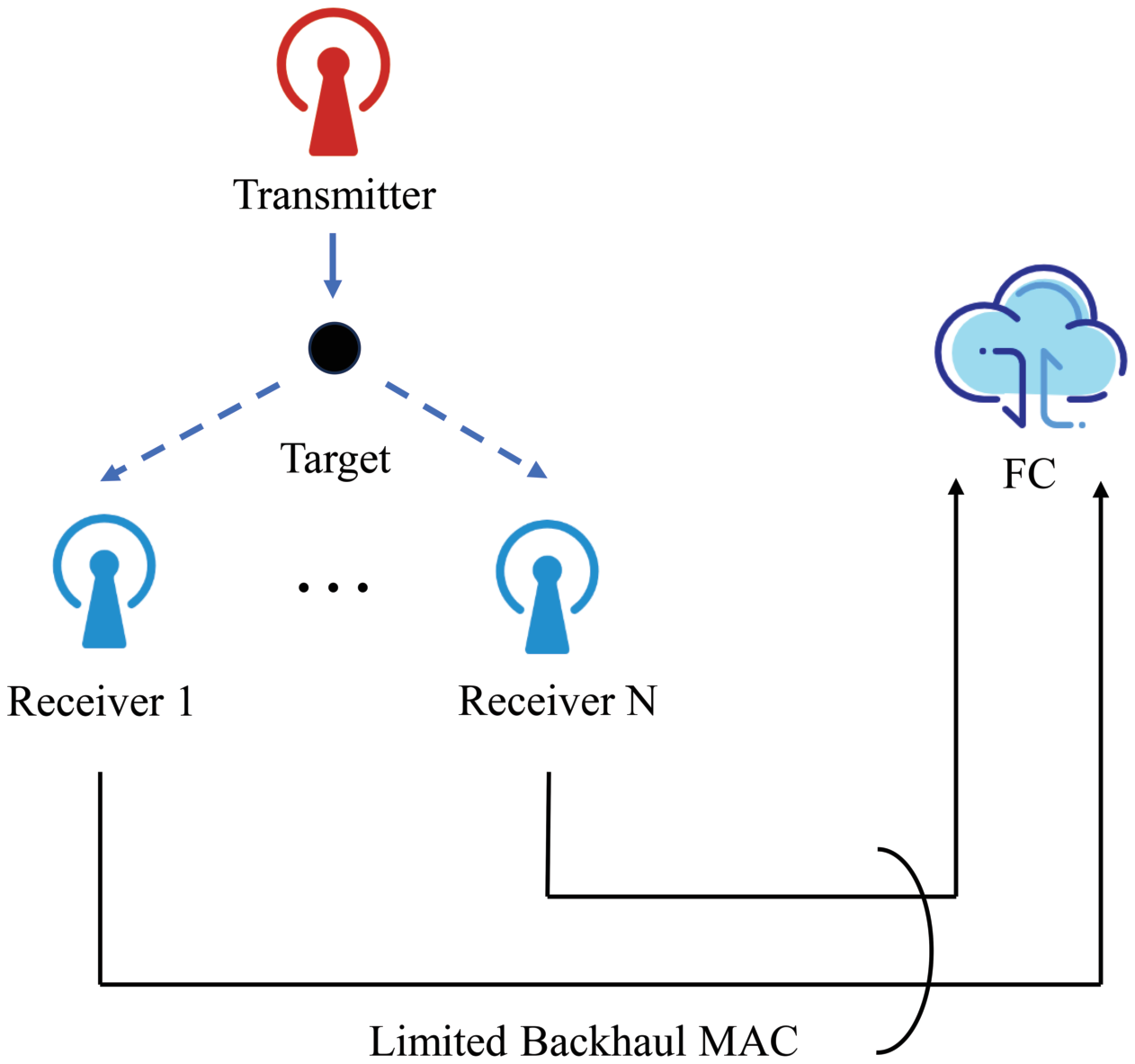}}
\caption{Cooperative sensing system model.}
\label{fig1}
\end{figure}

The rest of this paper is organized as follows. In Section II, we introduce the system model. Section III elaborates the HISDCS scheme. Section IV formulates an optimization problem and designs an efficient algorithm to solve the problem. Simulation results are provided in Section V and Section VI concludes the paper.

\emph{Notations}: Scalars, vectors and matrices are respectively denoted by lower/upper case, boldface lower case and boldface upper case letters. $\mathbb{E}(\mathbf{X})$ denotes the expectation of the random variable $\mathbf{X}$. $\Re(x)$ and $\Im(x)$ denote the real and imaginary parts of a complex number $x$, respectively. $\mathbf{A}^T$ and $\mathbf{A}^{-1}$ denote the transpose and inverse of matrix $\mathbf{A}$, respectively. $\mathbf{I}_{n}$ denotes a $n\times n$  identity matrix. $\operatorname{diag}(\mathbf{a})$ denotes a diagonal matrix with the elements in $\mathbf{a}$ being its diagonal elements. $\|\mathbf{a}\|$ denotes the $l_2$-norm of vector $\mathbf{a}$, and $\mathbf{A}\succeq\mathbf{B}$ means that $\mathbf{A}-\mathbf{B}$ is a positive semi-definite matrix.

\section{System Model}
Consider a cooperative sensing system with $N+1$ base stations (BSs), where one BS acts as the sensing transmitter while the other $N$ BSs act as the sensing receivers, as shown in Fig. \ref{fig1}. The $N$ sensing receivers are linked to the FC via a backhaul-limited Gaussian MAC. The sensing signal is sent by the transmitter, reflected by the target and then received by the $N$ receivers. After certain signal processing and local location information extraction, the $N$ receivers then quantify these local information and forward it to the FC. Finally, the FC estimates the target location based on the quantized information transmitted by the receivers.

Assuming that the transmitter and receivers are located at known positions $(x^t,y^t)$ and $(x^r_n, y^r_n),n \in [1,N]$, respectively, while the target is located at an unknown position, defined as $\boldsymbol{\theta} \triangleq [x,y]^T$. Besides, let $\sqrt{E}s(t)$ denote the lowpass equivalent of the signal transmitted from the transmitter, where $E$ is the transmitted power and $s(t)$ is a power normalized waveform satisfying $\int_{T_c}\left | s(t) \right | ^2dt=1$ with $T_c$ being the signal duration time. Then, the echo signal received at the $n$-th receiver is
\begin{equation}
r_n(t)=\sqrt{E}\alpha_ns(t-\tau_n)+\omega_n(t),\label{eq1}
\end{equation}
where $\omega_n(t)\sim \mathcal{CN} (0,\sigma^2_n)$ is the complex Gaussian noise at the $n$-th receiver, $\alpha_n =\rho _n  \xi _n$ is the effective reflecting coefficient of the $n$-th receiver, $\rho_n$ is the corresponding path-loss coefficient and $\xi _n$ is the reflection coefficient, which is assumed to be a complex random variable with unknown (deterministic) amplitude and random phase uniformly distributed between 0 and $2\pi$.  We assume that  the noises $\omega_n(t)$'s at different sensing receivers are statistically independent, due to the spatial separation, varying environmental conditions, and differing hardware factors, etc. This assumption is made to simplify the analysis without loss of generality and is commonly used in similar signal processing frameworks \cite{zhang2023direct,fishler2006spatial}.  Besides,  the direct signal transmitted from the transmitter to the receivers is not considered in \eqref{eq1} because in many radar and communication systems \cite{zhou2024direct,yang2022multitarget,wang2018parameter}, this direct signal usually arrives first with a different amplitude or phase as compared to the echo signal, thus it can be easily separated or filtered out from the reflected echo signals using techniques such as temporal filtering \cite{zhou2024direct}, thanks to the globally known positions of both the transmitter and the receivers.  In addition, the time delay $\tau_n$ corresponding to the $n$-th receiver can be expressed as
\begin{equation}
\tau_n = \frac{\sqrt{(x^t-x)^2+(y^t-y)^2} +\sqrt{(x^r_n-x)^2+(y^r_n-y)^2} }{c}, \label{eq2}
\end{equation}
where $c$ denotes the speed of light. After receiving the echo signals, the $N$ sensing receivers first sample the signals as
\begin{equation}
r_n(kT_s)=\sqrt{E}\alpha_ns(kT_s-\tau_n)+\omega_n(kT_s),k \in [1,K],
\end{equation}
where $T_s$ is the sampling period. Then, the $N$ sensing receivers quantize and process the sampled signals and communicate with the FC for cooperative sensing. Besides, it is worth noting that employing imprecise measurements from unreliable nodes may lead to localization accuracy deterioration \cite{hadzic2011utility}. Therefore, limiting the degree of cooperation and only using the measurements from the most informative nodes are essential for cooperative sensing.

Note that one baseline design for the considered problem is based on the IDCS scheme, where the $n$-th receiver first estimates the corresponding time delay $\tau_n$ and effective reflecting coefficient $\alpha_n$ from the received signal samples based on the ML rule and then transmits these estimated parameters to the FC. Finally, the FC estimates the target’s location by using the received parameters from the $N$ receivers by using the ML rule, again. Another baseline is based on the SDCS scheme, where each echo signal sample is uniformly quantized and then transmitted to the FC for localization \cite{wang2021target}.

The focus of this work is to combine the benefits of the IDCS and SDCS schemes by designing  proper local processing and quantization schemes at each receiver and efficient cooperative node selection strategy to guarantee the estimation performance of $\boldsymbol{\theta}$ at the FC while conserving the communication overhead.

\section{ Hybrid Information-Signal Domain Cooperative Sensing }

In order to enhance the sensing performance of the IDCS scheme while efficiently reduce the communication overhead of the SDCS scheme, a HISDCS scheme is proposed in this section by taking the limited MAC capacity into consideration. In the proposed design, each receiver first transmits the estimated time delay $\hat \tau_n$ and effective reflecting coefficient $\hat{\alpha}_n$ to the FC. Then, based on these estimated parameters, the FC solves a CRLB and MAC capacity constrained channel use number minimization problem to optimize the quantization bit allocation and node selection. Next, the FC sends the optimization results back to the $N$ receivers through dedicated links. The selected receivers sample their received sensing signals around the estimated time delay and employ the KLT encoding scheme \cite{goyal2001theoretical} to quantize these samples under the given quantization bit allocation. Finally, the quantized samples are transmitted to the FC for sensing performance improvement. The main structure of the proposed scheme is provided in Fig. \ref{fig2} and the key procedures mentioned above are detailed as follows.

\begin{figure}[htbp]
\centerline{\includegraphics[width=1\linewidth]{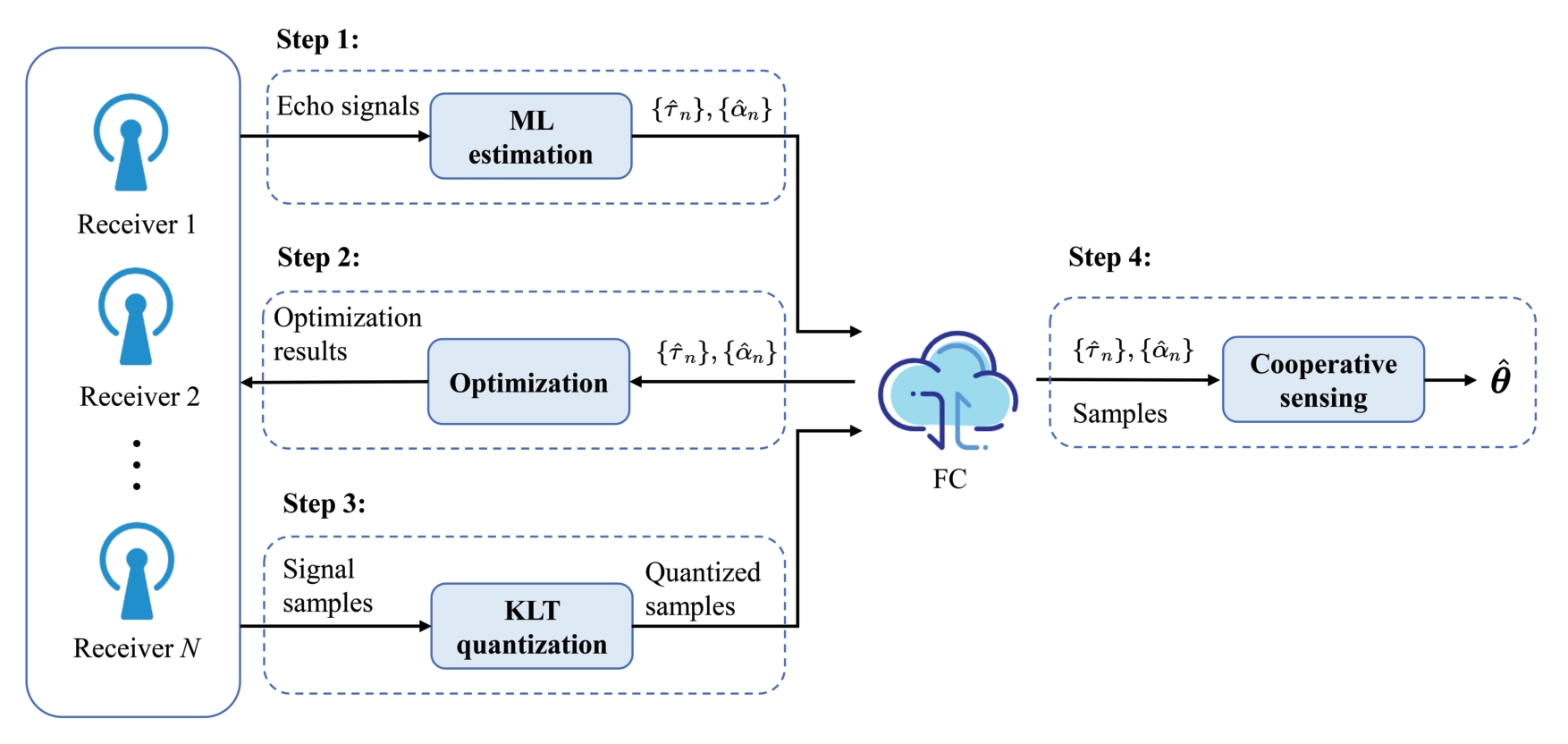}}
\caption{Illustration of the proposed design.}
\label{fig2}
\end{figure}

\subsection{ML Estimation of Individual Delay and Effective  Reflecting Coefficient}

Define $\mathbf{r} _n = \left \{ r_n(kT_s)|k=[1,K] \right \} $ as the  vector containing all the echo signals received by the $n$-th receiver, then the  probability density function (pdf) of $\mathbf{r} _n$ against $\tau_n$ and $\alpha_n$ in the log domain is given as
\begin{equation}
\begin{aligned}
&\ln p(\mathbf{r}_n|\tau_n,\alpha_n)=\\
&-\frac{1}{\sigma ^2_n} \sum_{k=1}^{K} \left | r_n(kT_s)-\sqrt{E}\alpha_ns(kT_s-\tau_n) \right |^2+D_0 ,
\label{eq4}
\end{aligned}
\end{equation}
where $D_0$ is a constant independent of $\tau_n$ and $\alpha_n$. Thus, the ML estimator for $\tau_n$ and $\alpha_n$ can be expressed as
\begin{equation}
[\hat{\tau}_n,\hat{\alpha}_n]=\arg \max_{\tau_n,\alpha _n}\ln p(\mathbf{r}_n|\tau_n, \alpha_n) .
\label{eq5}
\end{equation}

Note that directly solving \eqref{eq5} generally requires a two-dimensional search, which incurs high computational complexity. Thus, we propose to first obtain  a closed-form solution of  $\hat{\alpha}_n$ as a function of $\tau_n$, given by

\begin{equation}
\hat{\alpha}_n=\frac{ {\textstyle \sum_{k=1}^{K}}r_n(kT_s)s(kT_s-\tau_n) }{\sqrt{E}  {\textstyle \sum_{k=1}^{K}} \left | s(kT_s-\tau_n) \right |^2 } .
\label{eq6}
\end{equation}
Then, by substituting \eqref{eq6} back into \eqref{eq5}, the ML estimator for the time delay $\hat{\tau}_n$ can be obtained as
\begin{equation}
\hat{\tau}_n = \arg \max_{\tau_n} \left \{ -\frac{1}{\sigma^2_n}\sum_{k=1}^{K}\left | r_n\left ( kT_s \right )-\sqrt{E}\hat{\alpha}_ns\left ( kT_s-\tau_n \right )    \right | ^2   \right \} .
\label{eq7}
\end{equation}
Finally, once $\hat{\tau}_n$ is obtained from \eqref{eq7}, it can be substituted back into \eqref{eq6} to compute the corresponding estimate of $\hat{\alpha}_n$. The main procedure of the proposed ML algorithm is provided in Fig. \ref{fig_ML}:

\begin{figure}[htbp]
\centerline{\includegraphics[width=1.05\linewidth]{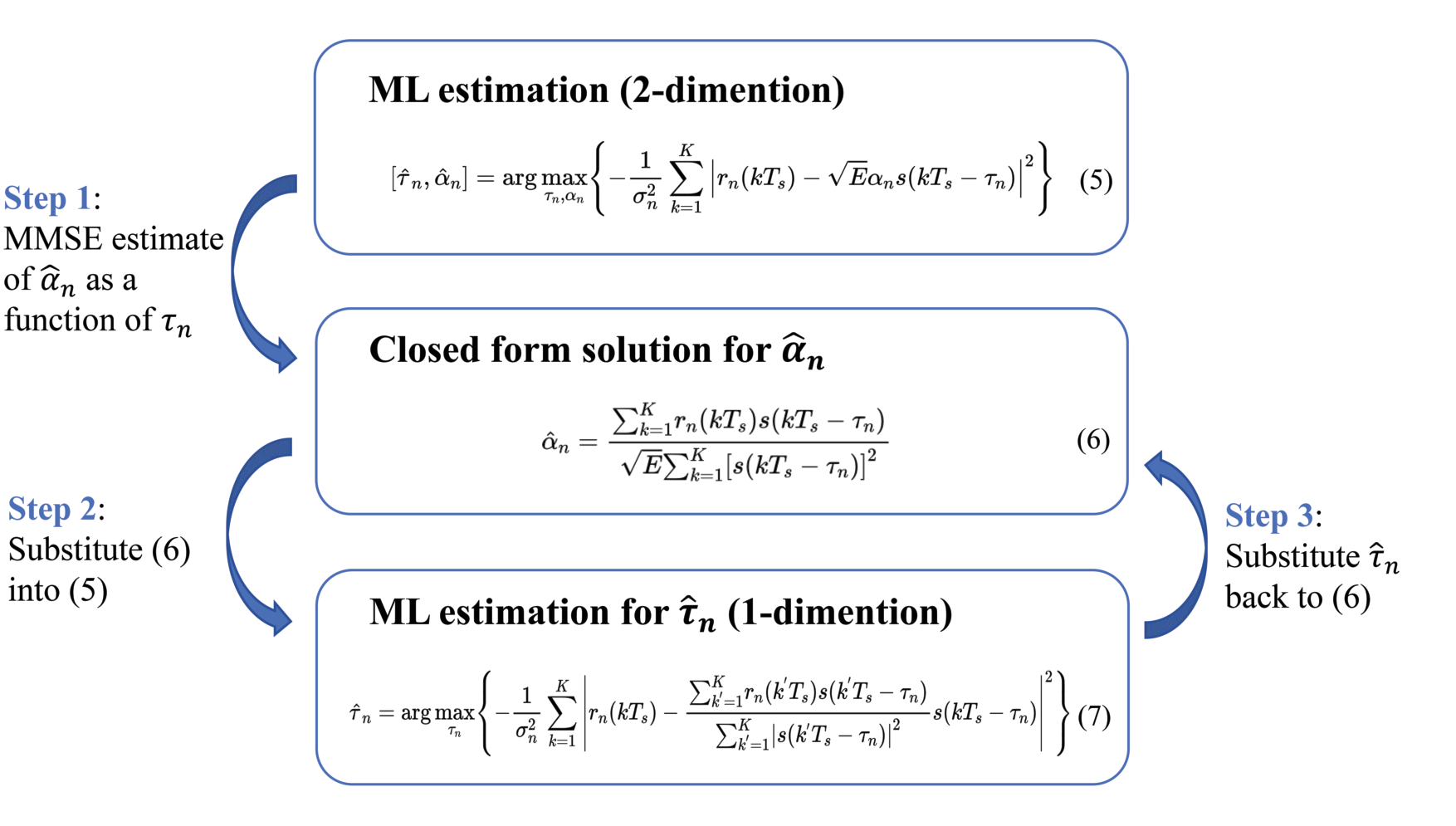}}
\caption{Procedure of the proposed ML algorithm.}
\label{fig_ML}
\end{figure}

The abovementioned approach ensures optimal estimation performance because the solution in equation \eqref{eq6} provides the minimum mean squared error (MMSE) estimate of $\alpha_n$ with a fixed $\tau_n$, and the maximum likelihood estimation method used in \eqref{eq5} is optimal when there is no prior information about the target. Additionally, this approach simplifies the estimation process by reducing the original two-dimensional search problem to a one-dimensional search problem over $\tau_n$.

Since the CRLB offers a theoretical and achievable lower bound of any unbiased estimator, and similar to the model  in \cite{bliss2014cooperative}, we assume that the estimated $\hat{\tau}_n$ and $\hat\alpha_n$ are statistically independent and related to the true parameters $\tau_n$ and $\alpha_n$, respectively,  as follows: 
\begin{equation}
    \hat{\tau}_n = \tau_n + \omega_{\tau_n},
    \label{eq8}
\end{equation}
\begin{equation}
    \hat\alpha_n=\alpha_n+\omega_{\alpha_n},
\end{equation}
where $\omega_{\tau_n}\sim \mathcal{N} \left ( 0,\text{CRLB}_{\hat{\tau}_n} \right ) $ is a Gaussian noise and its variance $\text{CRLB}_{\hat{\tau}_n}$ is the CRLB of  $\tau_n$ shown as follows\cite{li2023cooperative}:
\begin{equation}
    \text{CRLB}_{\hat{\tau}_n}=\frac{1}{\frac{2E}{\sigma^2_n} \left | \hat{\alpha}_n \right | ^2 \sum_{k=1}^{K} \left | {\frac{\partial s(t)}{\partial t}}\mid _{t=kT_s-\hat{\tau}_n}\right |^2}.
\end{equation}
Similarly,  $\omega_{\alpha_n}\sim \mathcal{CN} \left ( 0,\text{CRLB}_{\hat{\alpha}_n} \right ) $ is a complex Gaussian noise with variance $\text{CRLB}_{\hat{\alpha}_n}$, given by
\begin{equation}
    \text{CRLB}_{\hat{\alpha}_n}=\frac{1}{\frac{2E}{\sigma^2_n} \sum_{k=1}^{K} \left | s(kT_s-\hat{\tau}_n)\right |^2} ,
\end{equation}

\subsection{Received Signal Quantization at Each Receiver}

In this subsection, aiming to reduce the communication overhead (caused by sending all echo signal samples to the FC) and obtain high cooperation gain, we propose to send the received signals sampled around $\hat{\tau}_n$ to the FC. By transmitting the signal samples, we are able to retain the detailed time-domain information that is critical for more precise localization. With the additional time samples as in the proposed method, we are supposed to achieve better resolution and accuracy in localization. Suppose that each receiver has obtained the quantization bit allocation scheme from the FC (the details are provided in Section IV), the KLT based encoding scheme is thus employed to quantize these samples, which is shown to maximize the coding gain for Gaussian sources \cite{kay1993fundamentals}.

First, we introduce a sample index set at the $n$-th receiver, denoted by $\mathcal{K}_n=\left \{ k_n \in \mathcal{Z} \mid \hat{\tau}_n-\frac{T_d}{2} \le k_nT_s \le \hat{\tau}_n+\frac{T_d}{2}  \right \}  $, where $T_d$ is the sampling interval length and the size of $\mathcal{K}_n$ is $K_n$. Then, we can obtain a $2K_n \times 1$ vector which collects the real and imaginary parts of the signal samples at the $n$-th receiver, i.e.,
\begin{equation}
    \mathbf{r}_n =\left [ \left \{ \Re(r_n(k_nT_s)),k_n \in \mathcal{K}_n \right \} , \left \{ \Im(r_n(k_nT_s)),k_n \in \mathcal{K}_n \right \}\right ] ^T.
\end{equation}

Second, in order to efficiently quantize the sampled signals, it is crucial to ultilize the statistical information contained in $\mathbf{r}_n$. However, it is very difficult to directly analyze the exact distribution of  $\mathbf{r}_n$. To make the analysis tractable, we propose to regard $r_n(k_nT_s)$ as a binary function of $\tau_n$ and $\alpha_n$, denoted as $r_n(k_nT_s)=g_n(\tau_n,\alpha_n)$. Then, by considering its first-order Taylor expansion around $\hat{\tau}_n$ and $\hat\alpha_n$, we can approximate $r_n(k_nT_s)$ as follows:
\begin{equation}
\begin{aligned}
    r_n(k_nT_s)&=g_n(\tau_n,\alpha_n)\\
&\approx g_n(\hat{\tau}_n,\hat\alpha_n)+(\tau_n-\hat\tau_n)\cdot\frac{\partial g_n(\hat\tau_n,\hat\alpha_n)}{\partial \tau_n}\\
&\quad+(\alpha_n-\hat\alpha_n)\cdot\frac{\partial g_n(\hat\tau_n,\hat\alpha_n)}{\partial \alpha_n}\\
&=\sqrt{E}\omega_{\tau_n} \hat\alpha_n \frac{\partial s(t)}{\partial t} \mid _{t=k_nT_s-\hat{\tau}_n}\\
&\quad+\sqrt{E}(\hat\alpha_n-\omega_{\alpha_n})s(k_nT_s-\hat{\tau}_n) +\omega_n(k_nT_s).
\end{aligned}
\label{eq11}
\end{equation}
Under the approximation in \eqref{eq11},  the approximate value of $\mathbf{r}_n$ is obtained as
\begin{equation}
\begin{aligned}
 {\mathbf{r}}_n &\approx \sqrt{E}\omega_{\tau_n}\mathbf{p}_{\mathbf{r}_n}-\sqrt{E}\Re(\omega_{\alpha_n})\mathbf{q}_{\mathbf{r}_n1}-\sqrt{E}\Im(\omega_{\alpha_n})\mathbf{q}_{\mathbf{r}_n2}\\
 &\quad +\sqrt{E}\mathbf{h}_{\mathbf{r}_n}+\mathbf{w}_{\mathbf{r}_n},   
\end{aligned}
\end{equation}
where 
\begin{equation}
\begin{aligned}
    \mathbf{p}_{\mathbf{r}_n} =[\{ \Re(\hat\alpha_n\frac{\partial s(t)}{\partial t} \mid_{t=k_nT_s-\hat\tau_n}),k_n \in \mathcal{K}_n \},\\
    \{ \Im(\hat\alpha_n\frac{\partial s(t)}{\partial t} \mid_{t=k_nT_s-\hat\tau_n}),k_n \in \mathcal{K}_n \}]^T,
\end{aligned} 
\end{equation}
\begin{equation}
\begin{aligned}
    \mathbf{q}_{\mathbf{r}_n1} =[\{ \Re(s(k_nT_s-\hat\tau_n)),k_n \in \mathcal{K}_n \},\\
    \{ \Im( s(k_nT_s-\hat\tau_n)),k_n \in \mathcal{K}_n \}]^T,
\end{aligned}
\end{equation}
\begin{equation}
\begin{aligned}
    \mathbf{q}_{\mathbf{r}_n2} =[\{-\Im(s(k_nT_s-\hat\tau_n)),k_n \in \mathcal{K}_n \},\\
    \{ \Re( s(k_nT_s-\hat\tau_n)),k_n \in \mathcal{K}_n \}]^T,
\end{aligned}
\end{equation}
\begin{equation}
\begin{aligned}
    \mathbf{h}_{\mathbf{r}_n} =[\{ \Re(\hat\alpha_ns(k_nT_s-\hat\tau_n)),k_n \in \mathcal{K}_n \},\\
    \{ \Im(\hat\alpha_n s(k_nT_s-\hat\tau_n)),k_n \in \mathcal{K}_n \}]^T,
\end{aligned}
\end{equation}
\begin{equation}
\begin{aligned}
      \mathbf{w}_{\mathbf{r}_n} =[\{ \Re(\omega_n(k_nT_s),k_n \in \mathcal{K}_n \},\\
      \{ \Im(\omega_n(k_nT_s)),k_n \in \mathcal{K}_n \}]^T.
\end{aligned}
\end{equation}
Since we assume that $\omega_{\alpha_n} \sim \mathcal{CN}(0,\textrm{CRLB}_{\hat\alpha_n})$ is complex Gaussian distributed, its real and imaginary parts are independent and each follows Gaussian distribution $\mathcal{N}(0,\frac{1}{2}\textrm{CRLB}_{\hat\alpha_n})$. Then, it can be inferred that $\mathbf{r}_n$ also obeys Gaussian distribution since $\omega_n(k_nT_s)$, $\omega_{\tau_n}$, $\Re(\omega_{\alpha_n})$ and $\Im(\omega_{\alpha_n})$ are Gaussian distributed. Hence, the mean of $\mathbf{r}_n$ can be easily obtained as
\begin{equation}
     \bar{\mathbf{r}}_n = \sqrt{E}\mathbf{h}_{\mathbf{r}_n},
\end{equation}
and  the  covariance matrix $\mathbf{Q}_{\mathbf{r}_n}$ of $\mathbf{r}_n$ is given by
\begin{equation}
\begin{aligned}
\mathbf{Q}_{\mathbf{r}_n}&=\mathbb{E}[(\mathbf{r}_n-\bar{\mathbf{r}}_n)(\mathbf{r}_n-\bar{\mathbf{r}}_n)^T]\\ 
&=\mathbb{E}(\mathbf{r}_n\mathbf{r}_n^T)-\bar{\mathbf{r}}_n\bar{\mathbf{r}}_n^T\\ 
&=E\cdot\mathbf{p}_{\mathbf{r}_n}\mathbf{p}_{\mathbf{r}_n}^T\mathbb{E}(\omega_{\tau_n}^2)+ E\cdot\mathbf{q}_{\mathbf{r}_n1}\mathbf{q}_{\mathbf{r}_n1}^T\mathbb{E}(\Re(\omega_{\alpha_n})^2)\\ 
&\quad + E\cdot\mathbf{q}_{\mathbf{r}_n2}\mathbf{q}_{\mathbf{r}_n2}^T\mathbb{E}(\Im(\omega_{\alpha_n})^2)\\
&\quad +E\cdot\mathbf{h}_{\mathbf{r}_n}\mathbf{h}_{\mathbf{r}_n}^T + \mathbb{E}(\mathbf{w}_{\mathbf{r}_n}\mathbf{w}_{\mathbf{r}_n}^T)-\bar{\mathbf{r}}_n\bar{\mathbf{r}}_n^T\\ 
&=E\cdot \mathrm{CRLB}_{\hat\tau_n}\mathbf{p}_{\mathbf{r}_n}\mathbf{p}_{\mathbf{r}_n}^T + E\cdot \frac{1}{2}\mathrm{CRLB}_{\hat\alpha_n}\mathbf{q}_{\mathbf{r}_n1}\mathbf{q}_{\mathbf{r}_n1}^T\\
&\quad +E\cdot \frac{1}{2}\mathrm{CRLB}_{\hat\alpha_n}\mathbf{q}_{\mathbf{r}_n2}\mathbf{q}_{\mathbf{r}_n2}^T  + \frac{1}{2} \sigma^2_n\mathbf{I}_{2K_n}. 
\end{aligned}
\label{eq17}
\end{equation}%

Finally, based on the above approximation and statistical knowledge of $\mathbf{r}_n$, we can apply the KLT based encoding scheme to quantize $\mathbf{r}_n$ via the following two steps.

In the first step, we obtain $\mathbf{Q}_{\mathbf{r}_n}=\mathbf{U}_n\mathbf{\Lambda}_n \mathbf{U}_n^T$ via eigenvalue decomposition, where the diagonal matrix $\mathbf{\Lambda}_n=\textrm{diag}(\gamma_{n1},\gamma_{n2},\cdots,\gamma_{n(2K_n)})$ collects all the eigenvalues of $\mathbf{Q}_{\mathbf{r}_n}$. Then, by transforming $\mathbf{r}_n$ using KLT, we can obtain the following transformed vector: 
\begin{equation}
    \mathbf{r}_{nC}= \mathbf{U}_n^T\mathbf{r}_n.
    \label{eq13}
\end{equation}
Since the covariance matrix of the transformed vector $\mathbf{r}_{nC}$ is the diagonal matrix $\mathbf{\Lambda}_n$, all the elements in $\mathbf{r}_{nC}$ are independent of each other and satisfy
\begin{equation}
    \left [ \mathbf{r}_{nC}  \right ] _j \sim \mathcal{N} \left ( \left [ \mathbf{U}_n^T \bar{\mathbf{r}}_n \right ] _j ,\gamma_{nj} \right ), j\in[1,2K_n].
    \label{eq14}
\end{equation}

In the second step, a Lloyd quantizer \cite{cover1999elements} is applied to quantize $\left [ \mathbf{r}_{nC}  \right ] _j,j\in [1,2K_n]$. Since the locations of the $N$ receivers are different and the impacts of their received signals on sensing performance are distinct, $\left [ \mathbf{r}_{nC}  \right ] _j,j\in [1,2K_n]$ can be quantized using different numbers of quantization bits.  Nevertheless, the total number of available quantization bits is constrained by the MAC capacity. 

\subsection{Target Localization at the FC}

Suppose that the FC can obtain the estimated time delays $\{\hat{\tau}_n\}$ and effective reflecting coefficients $\{\hat{\alpha}_n\}$ (information-domain) at the receivers with negligible quantization error, then the eigenmatrices  $\{\mathbf{U}_n\}$ can be acquired with negligible quantization loss, which implicitly contains the information-domain signals $\{\hat{\tau}_n, \hat{\alpha}_n\}$.  After receiving the quantized signal samples $\tilde{\mathbf{r}}_{nC}$ (signal-domain) from the $N$ receivers, the FC is able to recover the target location by fusing these signal-domain and information-domain measurements, i.e., $\{\hat{\tau}_n, \hat{\alpha}_n\}$ and  $\{\left [ \tilde{\mathbf{r}}_{nC}  \right ] _j,j\in [1,2K_n]\}$.

Specifically, to evaluate the impact of signal quantization on the sensing performance, we introduce the Gaussian quantization error model \cite{park2014fronthaul}, based on which the $j$-th quantized signal at the $n$-th receiver, i.e.,  $[\tilde{\mathbf{r}}_{nC}]_j$, can be expressed as
\begin{equation}
    [\tilde{\mathbf{r}}_{nC}]_j=[\mathbf{r}_{nC}]_j+q_{nj},
\end{equation}
where $q_{nj}\sim \mathcal{N}(0,\eta_{nj}) $ is the additive Gaussian quantization error and $\eta_{nj} = \frac{\gamma_{nj}}{2^{2X_{nj}-1}}$ is the lower bound of the variance of $q_{nj}$ , which can be obtained by
\begin{equation}
    I ( \left [ \tilde {\mathbf{r} }_{nC} \right ]_j; \left [  {\mathbf{r} }_{nC} \right ]_j  )= \frac{1}{2}\textrm{log}_2 \left( \frac{\eta_{nj}+\gamma_{nj}}{\eta_{nj}} \right) \le X_{nj} ,
    \label{eq16}
\end{equation}
with $X_{nj}$ denoting the number of bits allocated for quantizing $[\mathbf{r}_{nC}]_j$ and $I ( \left [ \tilde {\mathbf{r} }_{nC} \right ]_j; \left [  {\mathbf{r} }_{nC} \right ]_j  ) $ denoting the mutual information between $[\tilde{\mathbf{r}}_{nC}]_j$ and $[\mathbf{r}_{nC}]_j$. According to the rate-distortion theorem \cite{cover1999elements},  the above quantization variance can be achieved by using the KLT based encoding scheme if \eqref{eq16} holds. 

Then, by performing inverse linear transformation on $\tilde{\mathbf{r}}_{nC}$, i.e., multiplying $\tilde{\mathbf{r}}_{nC}$ on the left by eigenmatrix   $\mathbf{U}_n$, we are able to recover the received signal samples by
\begin{equation}
    \tilde{\mathbf{r} }_n = \mathbf{U}_n\tilde{\mathbf{r} }_{nC}=\mathbf{r}_n+\mathbf{U}_n\mathbf{q}_n  ,
\end{equation}
where $\mathbf{U}_n\mathbf{q}_n \sim \mathcal{N}(\mathbf{0},\mathbf{Q}_n )  $ is the quantization error vector of $\mathbf{r}_n$ and $\mathbf{Q}_n$ is given by
\begin{equation}
    \mathbf{Q}_n = \mathbf{U}_n\textrm{diag}([\eta_{n1},\cdots,\eta_{n(2K_n)}]) \mathbf{U}_n^T.
    \label{eq18}
\end{equation}

Finally, the FC estimates the target location by resorting to the ML rule. Let  $\Omega \subseteq \left \{ 1,2,\cdots,N \right \} $ denote the selected receivers (the detailed selection strategy will be introduced in Section IV), then the received signal samples available at the FC can be represented as $\tilde{\mathbf{r} }_{\Omega}= \left \{ \tilde{\mathbf{r}}_n \mid n \in \Omega \right \} $. Since the noises $\omega_n(t)$  are uncorrelated across the sensing receivers, as noted in Section II. The ML estimation of $\boldsymbol \theta \triangleq[x,y]^T$ can be obtained by integrating both information-domain and signal-domains as follows:
\begin{equation}
    \hat{\boldsymbol \theta} = \arg \max_{\boldsymbol \theta}\ln p(\tilde{\mathbf{r}}_{\Omega} \mid \boldsymbol{\theta},\hat{\boldsymbol{\tau}},\hat{\boldsymbol{\alpha}}) ,
    \label{eq19}
\end{equation}
where
\begin{equation}
\begin{split}
    &\ln p(\tilde {\mathbf{r} }_{\Omega} \mid 
 \boldsymbol{\theta},\hat{\boldsymbol{\tau}},\hat{\boldsymbol{\alpha}}) = \\
    &-\frac{1}{2} \sum_{n \in \Omega}\left [ (\tilde{\mathbf{r} }_n-\mathbf{s}_n)^T \left ( \mathbf{Q}_{\omega_n  } +\mathbf{Q}_{n} \right) ^{-1}(\tilde{\mathbf{r} }_n-\mathbf{s}_n) \right] +D_2,
\end{split}
\label{eq20}
\end{equation}
and $D_2$ is a constant uncorrelated with $\boldsymbol{\theta}$, $\mathbf{Q}_{\omega_n} \triangleq\frac{1}{2}\sigma^2_n\mathbf{I}_{2K_n}$ is a diagonal matrix, $\hat{\boldsymbol{\tau}} \triangleq\{\hat\tau_n\mid n=1,\cdots,N\}$, $\hat{\boldsymbol{\alpha}} \triangleq\{\hat\alpha_n\mid n=1,\cdots,N\}$ and $\mathbf{s}_n  \triangleq  [ \{ \sqrt[]{E}\Re(\hat{\alpha}_ns(k_nT_s-\tau_n)) ,k_n \in \mathcal{K}_n  \} ,$ $  \{\sqrt[]{E}\Im(\hat{\alpha}_n s(k_nT_s-\tau_n)) ,k_n \in \mathcal{K}_n \}]^T$.

\begin{remark}
In \eqref{eq19} and \eqref{eq20}, it seems that only  signal-domain information $\tilde{\mathbf{r} }_{\Omega}$ is used for localization, however, it is important to mention that  the information-domain parameters also play a significant role in the overall localization process. Specifically, the term $\mathbf{Q}_n$ in \eqref{eq20} is derived from \eqref{eq18}, while the eigenmatrix $\mathbf{U}_n$ in \eqref{eq18} is obtained by performing an eigenvalue decomposition of the matrix $\mathbf{Q}_{\mathbf{r}_n}=\mathbf{U}_n\mathbf{\Lambda}_n \mathbf{U}_n^T$, which is computed based on the estimated time delays $\hat\tau_n$ and reflection coefficients $\hat\alpha_n$ as described in \eqref{eq17}. Besides, the term $\mathbf{s}_n$ in \eqref{eq20} also involves $\hat\alpha_n$ and $\hat\tau_n$, where $\hat \alpha_n$ affects the amplitude of $\mathbf{s}_n$ and $\hat\tau_n$ is implicitly  entailed in the set $\mathcal{K}_n $.
Thus, although the final position estimation relies on the quantized signals, the underlying information-domain parameters (time delays and reflection coefficients) are also involved in localization and they are crucial for shaping the quantized signals. Besides, we recognize that although we have integrated information-domain measurements into the signal-domain processing, the proposed method may not be able to perfectly exploit them. Investigating more efficient HISDCS schemes would be an interesting future research direction.
\end{remark}

\section{Problem Formulation and Algorithm Design}

In this section, we formulate an optimization problem to minimize the communication cost (i.e., the number of channel uses), under the CRLB constraint that ensures the sensing performance and the MAC capacity constraints between the sensing receivers and the FC. The number of quantization bits at each reciever and the set of selected receivers are jointly optimized by proposing a novel MCSCA algorithm and a greedy node selection strategy. Besides, we prove that the proposed MCSCA algorithm is guaranteed to converge to the set of Karush-Kuhn-Tucker (KKT) solutions, and a simple greedy bit reallocation algorithm is further designed for lower computational complexity.

\subsection{Optimization Problem Formulation}

In this work, we employ the well-known CRLB as the sensing performance metric, since it offers a theoretical lower bound on the best achievable estimation accuracy under given observations \cite{kay1993fundamentals}. Specifically, the CRLB of $\boldsymbol{\theta}$ at $\hat{\boldsymbol{\theta}}$ can be obtained as
\begin{equation}
    \textrm{CRLB}_{\hat{\boldsymbol\theta}}=\textrm{tr}[\mathbf{J}^{-1}(\hat{\boldsymbol\theta}) ],
    \label{eq21}
\end{equation}
where $\mathbf{J}(\boldsymbol{\theta})$ denotes the Fisher information matrix (FIM), which is given by
\begin{equation}
    \begin{aligned}
\mathbf{J}(\boldsymbol\theta) & =E_{\tilde{\mathbf{r}}_{\Omega} \mid \boldsymbol{ \theta}}\left\{-\frac{\partial^{2}}{\partial \boldsymbol\theta \partial \boldsymbol{\theta}^{T}} \ln p\left(\tilde{\mathbf{r}}_{\Omega} \mid \boldsymbol\theta\right)\right\} \\
& =\sum_{n \in \Omega} \frac{\partial \mathbf{s}_{n}}{\partial \boldsymbol\theta}\left(\mathbf{Q}_{\omega_{n}}+\mathbf{Q}_{n}\right)^{-1}\left(\frac{\partial \mathbf{s}_{n}}{\partial \boldsymbol\theta}\right)^{T}.
\end{aligned}
\label{eq22}
\end{equation}

Then, according to the Gaussian MAC capacity limits \cite{tse2005fundamentals}, the number of quantization bits $\{X_{nj}\}$  must satisfy the following conditions for each possible non-empty subset $\mathcal{S}$ of $\Omega$:
\begin{equation}
    \begin{aligned}
    \sum_{n \in \mathcal{S}} R_n= \sum_{n \in \mathcal{S}}\frac{\sum_{j=1}^{2 K_ n} X_{n j}}{W} \le \log \left(1+\frac{\sum_{n \in \mathcal{S}} P_{n} g_{n}}{N_{0}}\right), \\
    \textrm { for all } \mathcal{S} \subseteq \Omega,
    \end{aligned}
    \label{eq23}
\end{equation}
where $R_n=\frac{\sum_{j=1}^{2K_n}X_{nj}}{W}$ represents the transmission rate in the unit of bit per channel use (bpcu), $W$ is the number of channel uses \cite{tse2005fundamentals}, $N_0$ denotes the noise power, $P_n$ and $g_n$ represent  the transmit power at the $n$-th sensing receiver and the channel gain between the $n$-th receiver and the FC, respectively. Please refer to Appendix A for the derivation of \eqref{eq23}. Therefore, the considered optimization problem can be formulated as\footnote{Note that in this work, we assume that the $N$ receivers transmit signals using the same power, thus the transmit power $P_n$ is not an optimization variable.}
\begin{subequations}
    \begin{align} 
\mathop{\min}_{\Omega,{X_{nj}},W}  &W \label{eq24a}\\
 \textrm{s.t.} \quad & \eqref{eq23}, \notag \\
       & \text{CRLB}_{\hat{\boldsymbol\theta}} \le \epsilon, \label{eq24b}\\
       & W \in \mathcal{N}^+, \\
       &X_{nj} \in \mathcal{N},j \in [1,2K_n],n \in \Omega,
    \end{align}
    \label{eq24}%
\end{subequations}
where $W$ is minimized as the objective function to reduce the communication cost and constraint \eqref{eq24b}  is introduced to guarantee certain sensing performance. It is noteworthy that the feasibility of problem \eqref{eq24} can be ensured by setting a proper value of  $\epsilon$ that exceeds the minimum achievable CRLB (denoted by $\epsilon^*$), which can be easily obtained by letting $X_{nj}\rightarrow \infty$ and then calculating the corresponding CRLB through \eqref{eq21}.

Problem \eqref{eq24} is challenging to solve because 1) the optimization variables are discrete and intricately coupled in the constraints; 2) the CRLB constraint \eqref{eq24b} is highly non-convex which involves matrix inversion operation; and 3) selecting appropriate cooperative nodes is a combinational problem which is in general NP-hard. Generally, there is no efficient method for solving the non-convex problem \eqref{eq24} optimally. 

\subsection{Algorithm Design}

To address the above challenges, we propose in this work to decouple the optimization of the quantization bit numbers at each receiver $\{X_{nj}\}$, the required channel use number $W$, and the selected node set $\Omega$. Specifically,  we first present an efficient MCSCA algorithm to optimize $\{X_{nj}\}$ and $W$ with given $\Omega$, and then a greedy node selection strategy is presented to optimize $\Omega$ based on the MCSCA algorithm, which is able to find a high-quality solution of problem \eqref{eq24} effectively.

\subsubsection{MCSCA Algorithm for Quantization Bit Allocation}

With  given node selection, problem \eqref{eq24} is still difficult to solve due to the discrete variables $\{X_{nj}\}$ and $W$, and non-convex CRLB constraint. To tackle these challenges, we first slack the discrete variables $\{X_{nj}\}$ and $W$ into continuous ones. Then, in order to resolve the difficulty caused by the matrix inversion operation in CRLB,  we introduce a $2 \times 2$ auxiliary matrix $\mathbf{M}$ which satisfies
\begin{equation}
    \mathbf{M} \succeq \mathbf{J}^{-1}(\hat{\boldsymbol\theta}).
    \label{eq25}
\end{equation}
Since $ \mathbf{J}^{-1}(\hat{\boldsymbol\theta}) \succeq \mathbf{0}$, \eqref{eq25} can be equivalently transformed  into the following positive semidefinite constraint:
\begin{equation}
    \mathbf{G}=\begin{bmatrix}
 \mathbf{M}  & \mathbf{I}_2\\
 \mathbf{I}_2& \mathbf{J}(\hat{\boldsymbol\theta}) 
\end{bmatrix}
\succeq \mathbf{0},
\label{eq26}
\end{equation}
where $\mathbf{I}_2$ is a $2 \times 2$ identity matrix. Thus, according to above transformations, the quantization bit allocation sub-problem under given $\Omega$ can be obtained by
\begin{subequations}
    \begin{align} 
   \min_{\{X_{nj}\},W,\mathbf{M}}& W\\
 \textrm{s.t.}\quad  & \eqref{eq23}, \eqref{eq26}\notag\\
       & \textrm{tr}(\mathbf{M} ) \le \epsilon,\label{eq27b} \\
       & \{W ,X_{nj}\} \in \mathcal{R}^+,j\in[1,2K_n],n \in \Omega.\label{eq27c}
    \end{align}
    \label{eq27}
\end{subequations}

Next, since the eigenmatrix $\mathbf{U}_n$ is an orthogonal matrix, $\mathbf{J}(\hat{\boldsymbol{\theta}})$ can be equivalently rewritten as
\begin{equation}
\begin{aligned}
    &\mathbf{J}(\hat{\boldsymbol{\theta}})=\\
    &\sum_{n \in \Omega} \frac{\partial \mathbf{s}_{n}}{\partial \hat{\boldsymbol{\theta}}} \mathbf{U}_n \operatorname{diag}
    \left([y_n(X_{n1}),\cdots,y_n(X_{n(2K_n)})]\right)
    (\frac{\partial \mathbf{s}_{n}}{\partial \hat{\boldsymbol{\theta}}} \mathbf{U}_n)^{T}, 
\end{aligned}  
\label{eq28}
\end{equation}
where $y_n(X_{nj})$ is given by
\begin{equation}
    y_n(X_{nj})=\frac{2^{2X_{nj}-1}}{\gamma_{nj}+ 2^{2X_{nj}-2} \sigma^2_n} ,n \in \Omega,j \in[1,2K_n].
    \label{eq29}
\end{equation}
As can be seen, since $y_n(X_{nj})$ in \eqref{eq29} is a non-convex function with respect to $X_{nj}$, constraint \eqref{eq26} is a non-convex constraint and difficult to handle. Besides, since \eqref{eq26} is a matrix-inequality constraint, the existing successive convex approximation (SCA) algorithm \cite{razaviyayn2014successive} cannot be directly applied. Therefore, we propose the MCSCA algorithm to tackle the problem \eqref{eq27} effectively, where $\{X_{nj}\}$ are iteratively updated by solving a sequence of convex optimization problems.

Specifically, at the $t$-th iteration, a surrogate function $u^t_n(X_{nj})$ is constructed for each $y_n(X_{nj})$, which can be viewed as the tangent function of $y_n(X_{nj})$ at point $X^t_{nj}$ and is given by
\begin{equation}
    u^t_n(X_{nj})= (X_{nj}-X^t_{nj})\nabla y_n(X^t_{nj})+y_n(X^t_{nj}),
    \label{eq30}
\end{equation}
where 
\begin{equation}
    \begin{aligned}
        \nabla y_n(X^t_{nj})=\frac{\ln 2 \cdot 2^{2X^t_{nj}} \gamma_{nj}}{(\gamma_{nj}+ 2^{2X^t_{nj}-2})^2\sigma^2_n },\\
        y_n(X^t_{nj})=\frac{2^{2X_{nj}^t-1}}{\gamma_{nj}+ 2^{2X_{nj}^t-2} \sigma^2_n}.
    \end{aligned}
\end{equation}
Then, we can obtain a surrogate FIM  $\bar{\mathbf{J}}^t(\hat{\boldsymbol{\theta}})$ by substituting \eqref{eq30} into \eqref{eq28} and a surrogate matrix $\bar{\mathbf{G}}^t$ by replacing $\bar{\mathbf{J}}(\hat{\boldsymbol{\theta}})$ in \eqref{eq26} with $\bar{\mathbf{J}}^t(\hat{\boldsymbol{\theta}})$. Therefore, by employing a small positive number $\mu>0$, we have the following strongly convex problem:
\begin{subequations}
    \begin{align} 
   \mathop{\min}_{\mathbf{x},\{X_{nj}\},W,\mathbf{M}} & W+\mu\| \mathbf{x} - \mathbf{x}^t\|^2 \\
 \textrm{s.t.}\quad  &  \eqref{eq27c} ,\notag\\
       &  \textrm{tr}(\mathbf{M} )+\mu\|\mathbf{x} - \mathbf{x}^t\|^2 \le \epsilon,\\
       &  \bar{\mathbf{G}}^t=\begin{bmatrix}
 \mathbf{M}  & \mathbf{I}_2\\
 \mathbf{I}_2& \bar{\mathbf{J}}^t(\hat{\boldsymbol{\theta}}) 
\end{bmatrix}\succeq \mu\| \mathbf{x} - \mathbf{x}^t\|^2 \mathbf{I}_4,\label{32b}\\
       &  \sum_{n \in \mathcal{S}} \frac{\sum_{j=1}^{2 K_ n} X_{n j}}{W}+\mu\| \mathbf{x} - \mathbf{x}^t\|^2 \le  \\
       &\log \left(1+\frac{\sum_{n \in \mathcal{S}} P_{n} g_{n}}{N_{0}}\right),\textrm { for all } \mathcal{S} \subseteq \Omega \notag,\\
       & \left \| [X_{nj}-X^t_{nj}] \right \| ^2 \le (\beta^t)^2,\label{eq32c}
\end{align}
    \label{eq32}%
\end{subequations}
where $\mathbf{x}\triangleq[\{X_{nj}\},W,\text{vec}(\mathbf{M})^T]^T$ and the term $\mu\| \mathbf{x} - \mathbf{x}^t\|^2$ is added to ensure the strong convexity of the surrogate functions. Besides, $\{\beta^t\}$ in \eqref{eq32c} is a decreasing sequence satisfying $\beta^t \to 0$ and $\sum_t \beta^t = \infty$, which is introduced to gradually decrease the variable updating speed and ensure that the algorithm finally converges. Based on the above approximation, $\mathbf{x}^{t+1}$, $\{{X}_{nj}^{t+1}\}$ and $ W^{t+1}$ are obtained by solving \eqref{eq32} via existing software solvers, such as CVX \cite{guimaraes2015tutorial}.

The main steps of the proposed MCSCA algorithm are summarized in Algorithm 1, and it can be proved that Algorithm 1 is guaranteed to converge to the set of KKT solutions of problem \eqref{eq27}. The details will be presented in the following.

\makeatletter
\floatname{algorithm}{Algorithm}

  \renewcommand{\algorithmicrequire}{\textbf{Input:}}
  \renewcommand{\algorithmicensure}{\textbf{Output:}}
  
  \begin{algorithm}[H]
    \caption{Proposed MCSCA Algorithm for Quantization Bits Allocation}
    \begin{algorithmic}[1]
      \REQUIRE  $\Omega$,  $\{\hat{\tau}_n\}$, $\{\hat\alpha_n\}$, $\{\sigma^2_n\}$, $\{P_n\}$, $\{g_n\}$, $N_0$ and $\{\beta^t\}$. 
      \ENSURE  $\{\bar{X}_{nj}\}$ and $\bar{W}$.
      \STATE \textbf{Initialize} $\{X^0_{nj}\}$, $t=0$.
      \REPEAT
      \STATE Update the surrogate functions $u^t_n(X_{nj})$ according to \eqref{eq29};
      \STATE Solve problem \eqref{eq32} to obtain $\{{X}^{t+1}_{nj}\}$ and $ W^{t+1}$;
      \STATE $t=t+1$;
      \UNTIL some convergence criterion is met
    \end{algorithmic}
  \end{algorithm}

\begin{remark}

It is noteworthy that to make the overall problem tractable, we have ignored the discrete constraints and relaxed $\{X_{nj}\}$ as continuous variables in problem \eqref{eq27}. After obtaining the optimized quantization bit allocation, we can simply round up the continuous bit numbers to obtain a discrete solution, i.e.,
\begin{equation}
    X^*_{nj}=\left \lceil \bar X_{nj} \right \rceil , n \in \Omega,j\in[1,2K_n].
\end{equation}
and the final channel use number $W^*$ is set to be the minimum integer that satisfies the constraint \eqref{eq23}.
\end{remark}

\subsubsection{Convergence Analysis for the MCSCA Algorithm}
In order to facilitate our convergence analysis of the MCSCA algorithm, we propose to properly modify the problem formulation and algorithm design, which are detailed as follows.

First, the positive semidefinite constraint \eqref{eq26} is equivalent to the following constraint: 
\begin{equation}
    \lambda _{\text{min}}(\mathbf{G})\ge 0,
    \label{eq35}
\end{equation}
where  $\lambda _{\text{min}}(\mathbf{G})$ denotes the minimum eigenvalue of $\mathbf{G}$. Since the smallest eigenvalue of a symmetric matrix is equal to its minimum Rayleigh quotient, for $\forall \gamma \in [0,1]$, we have
\begin{equation}
    \begin{aligned}
&\lambda_{min}(\gamma \mathbf{G}+(1-\gamma)\mathbf{G}')\\
&=\min_{\boldsymbol{v}} \frac{\boldsymbol{v}^T(\gamma \mathbf{G}+(1-\gamma)\mathbf{G}')\boldsymbol{v}
}{\boldsymbol{v}^T\boldsymbol{v}}\\
&=\min_{\boldsymbol{v}}\left \{ \gamma\frac{\boldsymbol{v}^T\mathbf{G}\boldsymbol{v}}{\boldsymbol{v}^T\boldsymbol{v}} + (1-\gamma)\frac{\boldsymbol{v}^T\mathbf{G}'\boldsymbol{v}}{\boldsymbol{v}^T\boldsymbol{v}} \right \}\\
& \ge \gamma \min_{\boldsymbol{v}_1} \frac{\boldsymbol{v}_1^T\mathbf{G}\boldsymbol{v}_1}{\boldsymbol{v}^T_1\boldsymbol{v}_1}+(1-\gamma)\min_{\boldsymbol{v}_2} \frac{\boldsymbol{v}_2^T\mathbf{G}'\boldsymbol{v}_2}{\boldsymbol{v}^T_2\boldsymbol{v}_2}\\
&=\gamma \lambda_{min}(\mathbf{G})+(1-\gamma) \lambda_{min}(\mathbf{G}').
\end{aligned}
\end{equation}
Hence, $\lambda _{min}(\mathbf{G})$ is a concave function of $\mathbf{G}$, and problem \eqref{eq27} can be equivalently reformulated as follows:
\begin{subequations}
    \begin{align}
 \min_{\boldsymbol{x} \in\mathcal{X} } \quad &f_0(\boldsymbol{x})\\
 \textrm{s.t.} \quad &f_i(\boldsymbol{x})\le 0, i=1,\cdots,m,\label{eq37b}
    \end{align}
\label{eq37}%
\end{subequations}
where $\mathcal{X}$ is the domain of $\boldsymbol{x}\triangleq[\{X_{nj}\},W,\text{vec}(\mathbf{M})^T]^T$, $f_0(\boldsymbol{x})\triangleq W$ and $f_1(\boldsymbol{x}) \triangleq -\lambda _{min}(\mathbf{G})$. The rest constraints in \eqref{eq37b}, i.e., $f_i(\boldsymbol{x}) \le 0 ,i\in[2,m]$, can be easily obtained from \eqref{eq23} and \eqref{eq27b}, the details are omitted here for brevity.

Next, according to the structure of the MCSCA algorithm, in each iteration $t$, we replace the objective function and constraints $f_i(\boldsymbol{x}),i\in[0,m]$ by the following strongly convex surrogate functions $\bar f_i^t(\boldsymbol{x}),i\in\{0,\cdots,m\}$:
\begin{equation}
    \bar f_i^t(\boldsymbol{x})=
\begin{cases}-\lambda _{min}(\bar{\mathbf{G}}^t)+\mu\| \boldsymbol{x} - \boldsymbol{x}^t\|^2,&i=1,
 \\f_i(\boldsymbol{x})+\mu\| \boldsymbol{x} - \boldsymbol{x}^t\|^2,&i\in\{0\}\cup\{2,\cdots,m\},
\end{cases}
\label{eq38}
\end{equation}
where $\mu>0$ is a small positive number. Note that in \eqref{eq38} the convexity of $-\lambda _{min}(\bar{\mathbf{G}}^t)$ and $f_i(\boldsymbol{x}),i\in\{0\} \cup[2,m]$ can be easily established, and the term $\mu\| \boldsymbol{x} - \boldsymbol{x}^t\|^2$ is added to ensure the strong convexity of  $\bar f_i^t(\boldsymbol{x}),i\in[0,m]$.

Accordingly we solve the following surrogate convex problem the $t$-th iteration :
 \begin{subequations}
        \begin{align}
\boldsymbol{x}^{t+1}=\arg \min_{\boldsymbol{x} \in\mathcal{X} } \quad &\bar f_0^t(\boldsymbol{x})\\
\textrm{s.t.} \quad & \bar f_i^t(\boldsymbol{x})\le 0, i=1,\cdots,m,\\
& \left \| \boldsymbol{x}-\boldsymbol{x}^t \right \| ^2 \le (\beta^t)^2.\label{eq39c}
\end{align}
\label{eq39}%
  \end{subequations}

Note that although only the elements $\{X_{nj}\}$ in $\boldsymbol{x}$ are required to satisfy the constraint \eqref{eq39c}, we impose this constraint on all the optimization variables in $\boldsymbol{x}$, which will simplify the convergence proof, but will not affect the performance of the proposed MCSCA algorithm (as validated via numerical simulations). 

Then, we introduce some lemmas which are crucial for our proof.

\begin{lemma}
Consider the following optimization problem which is obtained by removing the constraint \eqref{eq39c} from problem \eqref{eq39}:
\begin{subequations}
    \begin{align}
    \bar{\boldsymbol{x}}^{t}=\arg \min_{\boldsymbol{x} \in\mathcal{X} } \quad &\bar f_0^t(\boldsymbol{x})\\
    \textrm{s.t.} \quad & \bar f_i^t(\boldsymbol{x})\le 0, i=1,\cdots,m.
    \end{align}
\label{eq40}%
\end{subequations}
Let $\{\boldsymbol{x}^{t}\}^{\infty }_{t=1}$ and $\{\bar{\boldsymbol{x}}^{t}\}^{\infty }_{t=1}$ denote the sequences of iterates generated by the proposed MCSCA algorithm when solving problems \eqref{eq39} and \eqref{eq40}, respectively, then we have
  \begin{equation}
      \lim_{t \to \infty} \left \| \bar{\boldsymbol{x}}^t-\boldsymbol{x}^t \right \| =0.
      \label{eq41}
  \end{equation}
  
\end{lemma}
\begin{proof}
Please refer to Appendix B for the proof.
\end{proof}

\begin{lemma}
\cite{liu2019stochastic} Consider a subsequence $\{\boldsymbol{x}^{t_j}\}^{\infty }_{j=1}$  converging to a limit point $\boldsymbol{x}^*$. There exist continuous functions $\hat f_i(\boldsymbol{x}),i=1,\cdots,m$ that satify
    \begin{equation}
        \lim _{j \rightarrow \infty} \bar{f}_{i}^{t_{j}}(\boldsymbol{x})=\hat{f}_{i}(\boldsymbol{x}), \forall \boldsymbol{x} \in \mathcal{X},
        \label{eq42}
    \end{equation}
    with the Slater's condition satisfied at $\boldsymbol{x}^*$ if there exist $\boldsymbol{x} \in \textrm{relint}\mathcal{X}$ such that 
    \begin{equation}
        \hat{f}_i(\boldsymbol{x})<0,\forall i =1,\cdots,m.
    \end{equation}
\end{lemma}

Based on Lemmas 1 and 2, we are ready to prove the following convergence result.
\begin{theorem}
Let $\boldsymbol{x}^0$ denote an initial feasible point, the limiting point $\boldsymbol{x}^*$ of $\{\boldsymbol{x}^{t}\}^{\infty }_{t=1}$ which satisfies the Slater's condition  is a stationary solution of problem \eqref{eq35}. 
\end{theorem}

\begin{proof}
According to Lemma 1, the constraint \eqref{eq39c} is inactive when $t \to \infty$, which implies that problem \eqref{eq39} is equivalent to problem \eqref{eq40} when $t$ is sufficiently large. Therefore, based on  \eqref{eq41} and \eqref{eq42}, we have
 \begin{subequations}
        \begin{align}
{\boldsymbol{x}}^{*}=\arg \min_{\boldsymbol{x} \in\mathcal{X} } \quad &\hat f_0(\boldsymbol{x})\\
\textrm{s.t.} \quad & \hat f_i(\boldsymbol{x})\le 0, i=1,\cdots,m,
\end{align}
\label{eq44}%
  \end{subequations}
as $t \to \infty$ and since the Slater's condition is satisfied, the KKT conditions of problem \eqref{eq44} indicates that there exist Lagrange dual variables $\lambda_1,\cdots,\lambda_m$ such that
\begin{equation}
    \begin{aligned}
\nabla \hat{f}_{0}\left(\boldsymbol{x}^{*}\right)+\sum_{i} \lambda_{i} \nabla \hat{f}_{i}\left(\boldsymbol{x}^{*}\right) & =\mathbf{0},\\
\hat{f}_{i}\left(\boldsymbol{x}^{*}\right) & \leq 0, \forall i=1, \ldots, m, \\
\lambda_{i} \hat{f}_{i}\left(\boldsymbol{x}^{*}\right) & =0, \forall i=1, \ldots, m .
\end{aligned}
\label{eq45}
\end{equation}
Then, owing to the fact that the surrogate function $u^t(X_{nj})$ is the tangent function of $y^t(X_{nj})$ according to \eqref{eq29}, and using \eqref{eq42}, we have the following facts:
\begin{equation}
\begin{aligned}
     \|\nabla\hat{f}_{i}(\boldsymbol{x}^*)-\nabla {f}_{i}(\boldsymbol{x}^*)\|=0,\forall i=1, \ldots, m .\\
      |\hat{f}_{i}(\boldsymbol{x}^*)- {f}_{i}(\boldsymbol{x}^*)|=0,\forall i=1, \ldots, m .
\end{aligned}
    \label{eq46}
\end{equation}

Finally, it follows from \eqref{eq45} and \eqref{eq46}  that $\boldsymbol{x}^*$ also satisfies the KKT conditions of the original problem \eqref{eq37}. Hence, $\boldsymbol{x}^*$ is a stationary point of problem \eqref{eq37}. This completes the proof.
\end{proof}

\subsubsection{Node Selection Strategy}

For node selection, we propose in this work a greedy strategy to iteratively exclude the most non-informative node from cooperation, such that the communication cost can be further reduced.

Specifically, we first initialize the selected node set as $\Omega^0 = \{1,2,\cdots,N\}$ by taking all the nodes into consideration. In each iteration $t_1$, Algorithm 1 is executed to obtain the optimized quantization bit numbers $\{\bar X^{t_1}_{nj}\}$ and channel use number $\bar W^{t_1}$ according to the current set of selected nodes $\Omega^{t_1}$. Then, in the next iteration $t_1+1$, the node with the minimum number of allocated quantization bits, denoted by $n^{t_1}$, is excluded from $\Omega^{t_1}$ because generally this node exhibits the least contribution to localization performance improvement.  Consequently, a new selected node set is obtained by $\Omega^{t_1+1}=\Omega^{t_1} \setminus  n^{t_1}$. Next, Algorithm 1 is executed again to obtain the new optimized quantization bit numbers $\{\bar X^{t_1+1}_{nj}\}$ and channel use number $\bar W^{t_1+1}$.  If $\bar W^{t_1+1}> \bar W^{t_1}$, i.e., excluding the current node $n^{t_1}$ does not lead to a reduction in the number of channel uses, then the algorithm is terminated; otherwise, the above greedy node selection process continues. It is worth noting that in the above iterative process, Algorithm 1 may not be feasible, this indicates that the current node selection scheme cannot satisfy the CRLB constraint. In this case, the proposed algorithm should also be terminated. The above strategy is summarized in Algorithm 2.

\makeatletter
\floatname{algorithm}{Algorithm}

  \renewcommand{\algorithmicrequire}{\textbf{Input:}}
  \renewcommand{\algorithmicensure}{\textbf{Output:}}
  
  \begin{algorithm}[H]
    \caption{Proposed Greedy Node Selection Strategy}
    \begin{algorithmic}[1]
      \REQUIRE  $\{\hat{\tau}_n\}$ , $\{\hat\alpha_n\}$,  $\{\sigma^2_n\}$,  $\{B_n\}$,  $\{P_n\}$ , $\{g_n\}$ and $N_0$.
      \ENSURE  $\{{X}^*_{nj}\}$ , ${W}^*$ and $\Omega^*$.
      \STATE \textbf{Initialize} $\Omega^0 = \{1,2,\cdots,N\}$, $t_1=0$, $W^{-1}=+\infty$.
      \REPEAT
      \STATE Execute Algorithm 1 to obtain $\{\bar X^{t_1}_{nj}\}$,  $\bar W^{t_1}$ and $n^{t_1}$ according to  $\Omega^{t_1}$;
      \STATE Update the set of selected nodes by $\Omega^{t_1+1}=\Omega^{t_1} \setminus  n^{t_1}$;
      \STATE $t_1=t_1+1$;
      \UNTIL ($\bar W^{t_1}> \bar W^{t_1-1}$) or the Algorithm 1 is infeasible
      \STATE $\Omega^*=\Omega^{t_1-1}$, $X^*_{nj}=\bar X^{t_1-1}_{nj}$ and $W^* = \bar W^{t_1-1}$;
    \end{algorithmic}
  \end{algorithm}

It is important to note that the overall localization performance is influenced not only by the number of quantization bits but also by the network topology. In the proposed algorithm, the geometric configuration of the sensing nodes is taken into consideration through the calculation of the CRLB as described in \eqref{eq21} and \eqref{eq22}.

Next, we endeavor to analyze the computational complexity of the proposed greedy node selection algorithm. It is imperative to note that within each iteration of Algorithm 2, the proposed MCSCA algorithm is executed once. 
Hence, the complexity of Algorithm 2 is dominated by  the complexity of the MCSCA algorithm.
Besides, in each iteration of the MCSCA algorithm, problem \eqref{eq32} is solved by resorting to the  interior-point method (IPM), whose complexity can be expressed as \cite{wang2014outage}
\begin{equation}
    C_{IPM}=n\sqrt{4+2^N} [(4^3+2^N)+n(4^2+2^N)+n^2],
\end{equation}
where $n=2K_nN+5$. Let $L$ denote the average iteration number of the MCSCA algorithm, then the complexity order of the MCSCA algorithm can be obtained as
\begin{equation}
    C_1=O(2^{3N/2} N^2K_n^2L ).
\end{equation}
Next, note that in the worst-case scenario, our greedy strategy executes the MCSCA algorithm $N$ times, the overall worst-case complexity order of Algorithm 2  is given by
\begin{equation}
    C_{2}=O\left(\sum_{i=1}^{N}2^{3i/2}i^2K_n^2L\right)=O(2^{3N/2}N^2K_n^2L ).
\end{equation}
It is noteworthy that although the complexity of the proposed algorithm seems to be exponential with respect to $N$, the actual complexity will not be excessively high since the value of $N$ will not be very large in practice.

To summarize,  although the greedy approach of removing nodes with the fewest allocated bits is not  optimal, we believe it offers a reasonable balance between performance and complexity within the current scope of our work.

\makeatletter
\floatname{algorithm}{Algorithm}

  \renewcommand{\algorithmicrequire}{\textbf{Input:}}
  \renewcommand{\algorithmicensure}{\textbf{Output:}}
  
  \begin{algorithm}[t]
    \caption{Proposed Low-Complexity Bit Reallocation Algorithm}
    \begin{algorithmic}[1]
      \REQUIRE  $\{\hat{\tau}_n\}$ , $\{\hat\alpha_n\}$,  $\{\sigma^2_n\}$,  $\{B_n\}$,  $\{P_n\}$ , $\{g_n\}$ and $N_0$.
      \ENSURE $\{{X}^*_{nj}\}$ , ${W}^*$ and $\Omega^*$.
      \STATE \textbf{Initialize} $\text{CRLB}_{opt}=+\infty$ and $\Omega^*=\{1,2,\cdots,N\}$.
      \STATE Execute Algorithm 1 based on $\Omega^*$ to obtain $\{\bar{X}_{nj}\}$;
      \STATE Set $X_{nj}=\left \lceil \bar{X}_{nj} \right \rceil,n \in \Omega, j\in[1,2K_n]$;
      \WHILE{$\text{CRLB}_{opt} > \epsilon$}
      \FOR{$n$ = 1 : $N$}
      \FOR{$j=1:2K_n$}
      \STATE $X_{nj}=X_{nj}+1$;
      \STATE Calculate $\text{CRLB}_{\hat{\boldsymbol\theta}}$ according to \eqref{eq21};
      \IF{$\text{CRLB}_{\hat{\boldsymbol\theta}} \le \text{CRLB}_{opt}$}
      \STATE $\text{CRLB}_{opt}=\text{CRLB}_{\hat{\boldsymbol\theta}}$;
      \STATE $n'=n$, $j'=j$;
      \ENDIF
      \STATE $X_{nj}=X_{nj}-1$;
      \ENDFOR
      \ENDFOR
      \STATE $X_{n'j'}=X_{n'j'}+1$;
      \ENDWHILE
      \STATE Set $X_{nj}^*=X_{nj}$;
      \FOR{$n=1:N$}
      \IF{$\sum_j X_{nj}=0$ }
      \STATE $\Omega^*=\Omega^*\setminus n$;
      \ENDIF
      \ENDFOR
      \STATE Set $W^*$ to be the minimum integer that satisfies \eqref{eq23};
\end{algorithmic}
  \end{algorithm}

\subsubsection{Low-Complexity Bit Reallocation Algorithm}
The complexity of the above node selection algorithm is relatively high since Algorithm 1 is required to be executed multiple times. To reduce its computational complexity, we develop a bit reallocation idea and further propose a low-complexity algorithm in the following.

Specifically, we first initialize the selected node set as $\Omega = \{1,2,\cdots,N\}$ and execute the proposed MCSCA algorithm once to obtain the optimized quantization bit numbers $\{\bar X_{nj}\}$ (before rounding down). Then, a discrete solution can be obtained by performing the floor operation on each $\bar{X}_{nj}$, i.e., $X_{nj}=\left \lceil \bar{X}_{nj} \right \rceil,n \in \Omega, j\in[1,2K_n]$. Next, we propose to iteratively increase the quantization bit number $X_{n' j'}$ by one bit, where $n'$ and $j' $ is obtained by  $[n',j']=\arg\min_{n,j} \text{CRLB}_{\hat{\boldsymbol\theta}}|_{X_{nj}=X_{nj}+1}$. In other words, we choose the most important signal sample among all the receivers, i.e., increase the quantization bit number of this signal sample can reduce the CRLB most significantly. Repeat this process until the resulting CRLB satisfies the constraint \eqref{eq24b}. The final channel use number $W^*$ is set to be the minimum integer that satisfies the MAC capacity constraint \eqref{eq23}. The details are presented in Algorithm 3.

In terms of computational complexity,  we can see the complexity orders of Algorithms 2 and 3 are the same. However, it is straightforward to see that the actual complexity of Algorithm 3 is much lower than that of Algorithm 2, since the proposed MCSCA algorithm only needs to be run once in Algorithm 3.

\section{Simulation Results}

\begin{figure}[t!]
\centerline{\includegraphics[width=1\linewidth]{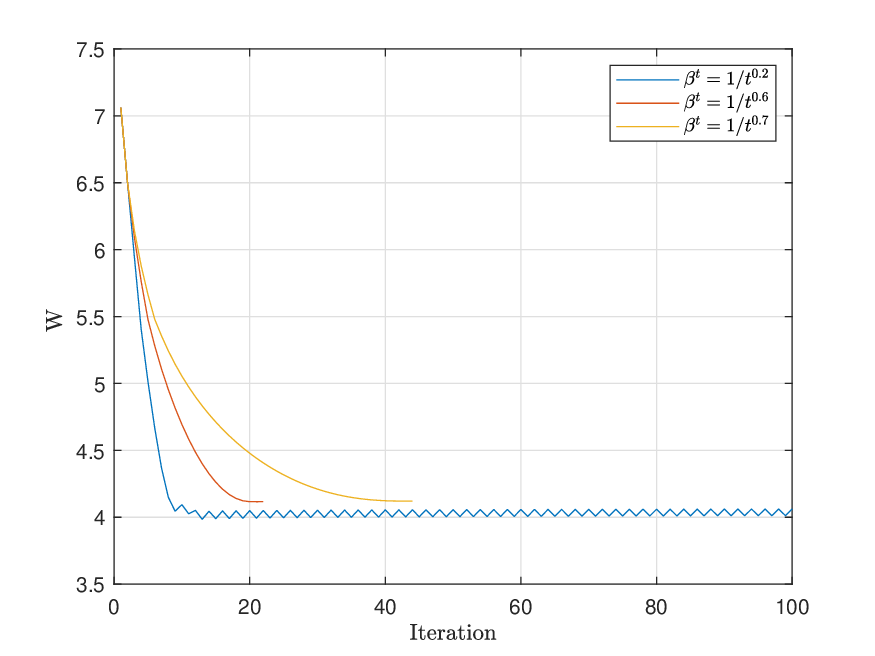}}
\caption{Convergence behavior of the proposed MCSCA algorithm.}
\label{fig3}
\end{figure}

In this section, we provide numerical results to evaluate
the performance of the proposed HISDCS scheme and draw useful insights. Without loss of generality, we  consider a common linear topology of the sensing receivers, as illustrated in Fig. \ref{fig1}. Unless otherwise specified, $N=5$ sensing receivers are considered, and they are located equidistantly in a straight line with 50 m spacing. The transmitter is located 1000 m away from this line. We assume that the target is located between the transmitter and receivers, uniformly distributed in the 100 m $\times$ 50 m region 50 m away from the line of sensing receivers. In our simulations, the sensing signal is set to be a Gaussian pulse signal, i.e., $s(t)=\frac{2^{0.25}}{T^{0.5}} e^{\frac{-\pi t^2}{T^2} }$ with $T=2\times 10^{-8}$ s. The carrier frequency, bandwidth and Nyquist sampling period are $f_c =3.55$ GHz, $B=50$ MHz and $T_s = \frac{1}{2B}=10^{-8}$ s, respectively. Considering that the main lobe of $s(t)$ is in $-1.5T \le t \le 1.5T$, the sampling duration is set to $T_d = 4T = 8 \times 10^{-8}$ s , which is relatively large to contain the main lobe of $s(t)$. Therefore, the number of sampling points within one pulse for the proposed scheme is $K_n = \frac{T_d}{T_s}+2=10$. 
Besides, the line-of-sight (LOS) model is employed to model the channels from the transmitter to the target, from the target to the receivers, and from the receivers to the FC, which is given by\footnote{The sensing signal transmitted might propagate through a primary LOS path along with some NLOS paths before reaching the receiver. Normally, only the LOS path is exploited for localization, while the NLOS paths can be treated as a part of the clutter.} \cite{rappaport2017overview}

\begin{equation}
    L=32.4+20\log_{10}(d(\textrm{km}))+20\log_{10}(f_c(\textrm{GHz}))(\textrm{dB}),
\end{equation}
where $L$ is the squared pathloss coefficient in dB, $d$ is the distance between nodes. 

\begin{figure}[t!]
\centering
\subfloat[MSE performance comparison under different SNR regimes]{\includegraphics[width=1\linewidth]{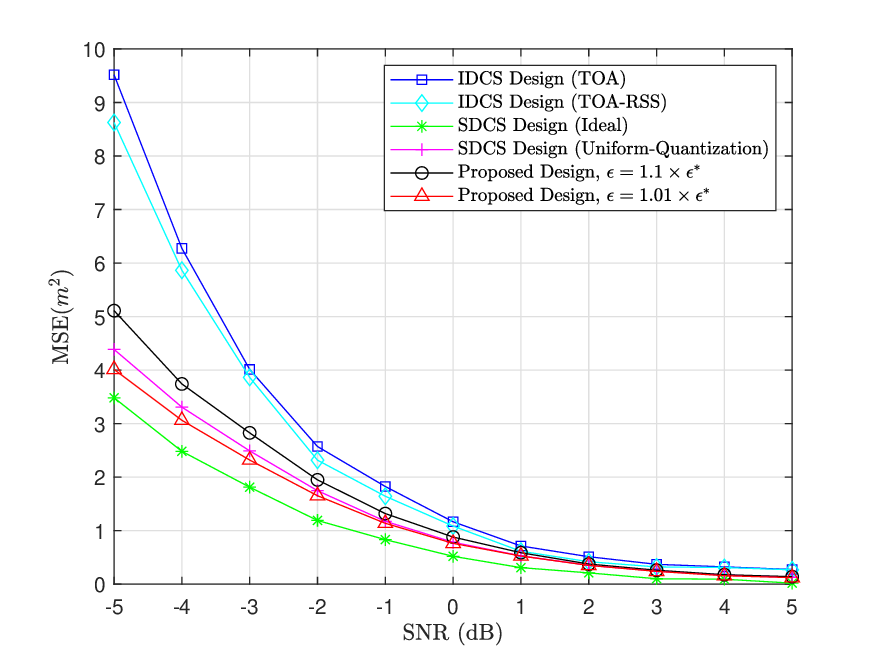}}%
\hfil
\subfloat[Channel use number comparison under different SNR regimes]{\includegraphics[width=1\linewidth]{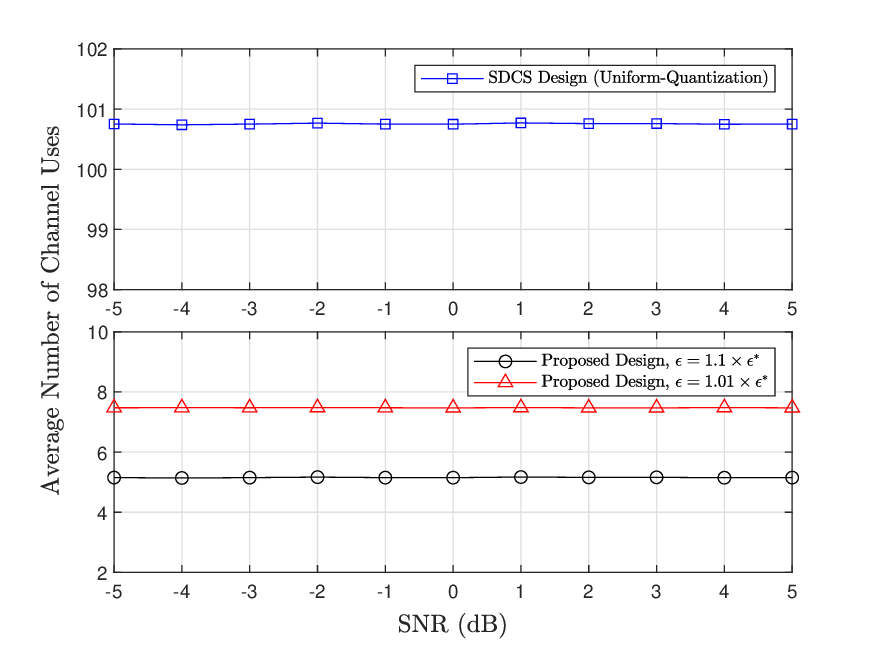}}%
\caption{Comparison between the proposed and baseline designs under different SNR regimes.}
\label{fig_snr_scheme}
\end{figure}

For comparison, the following three baselines are considered: 1) A TOA based IDCS scheme, where each receiver only sends the estimated time delay $\hat\tau_n$ and effective reflecting coefficient $\hat\alpha_n$ to the FC for target localization based on the ML rule; 2) A TOA-RSS based IDCS scheme, where the estimated time delay $\hat\tau_n$ , effective reflecting coefficient $\hat\alpha_n$ , and the signal power around $\hat \tau_n$ are  transmitted to the FC for localization; 3) A uniform-quantization SDCS scheme, where each echo signal sample is uniformly quantized using 8 bits and then transmitted to the FC for localization; 4) An ideal SDCS scheme with unlimited communication capacity, i.e., assuming perfect echo signals are available at the FC.

\begin{figure}[t!]
\centering
\subfloat[MSE performance comparison under different $N$]{\includegraphics[width=1\linewidth]{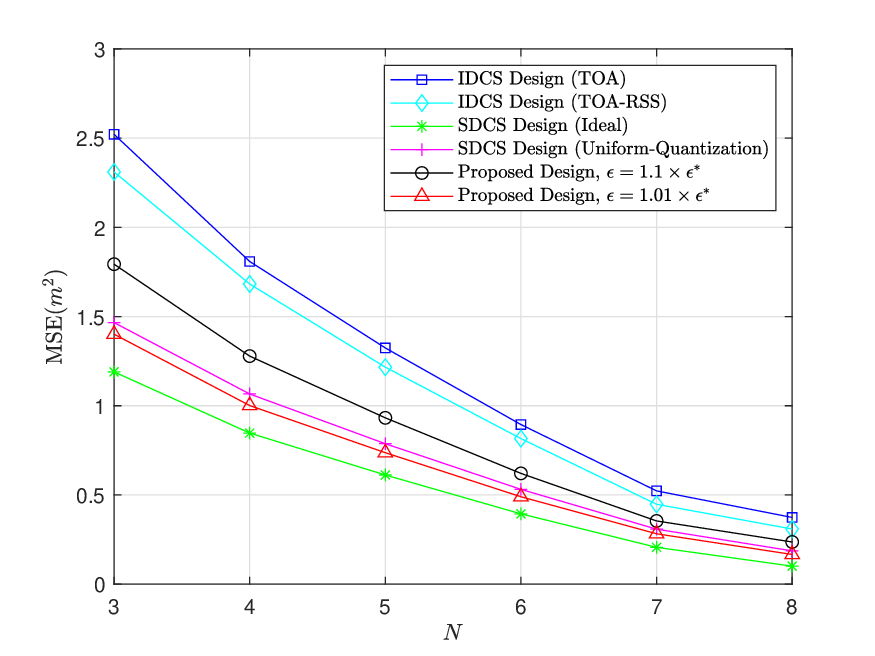}}%
\hfil
\subfloat[Channel use number comparison under different $N$]{\includegraphics[width=1\linewidth]{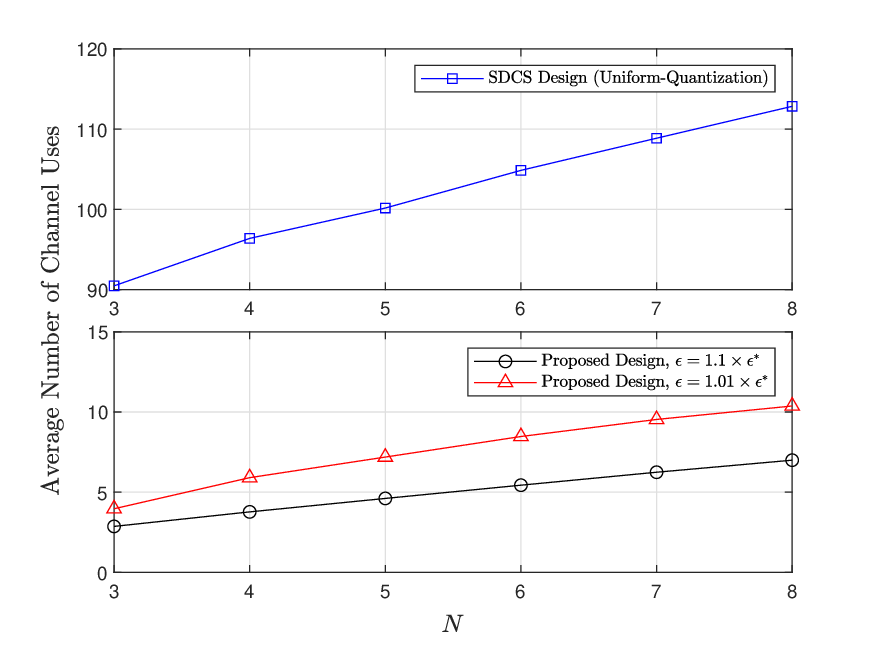}}%
\caption{Comparison between the proposed and baseline designs under different numbers of sensing receivers.}
\label{fig_n_scheme}
\end{figure}

First, Fig. \ref{fig3} plots the convergence behavior of the MCSCA algorithm under different choices of the step size  $\beta^t$. It is shown that by choosing a feasible initial point and a proper step size, the proposed MCSCA algorithm is able to converge within 20 iterations. Besides, we can see that choosing a proper step size is very important for the proposed algorithm, since a small step size can result in a slow convergence rate, while selecting a large step size may induce oscillations.

Then, we compare in Fig. \ref{fig_snr_scheme}  the MSE performance and the average channel use number $\overline{W} $ between the baseline designs and the proposed HISDCS design under different SNR regimes. It can be observed that in the medium-to-low SNR regime, the proposed HISDCS scheme outperforms the TOA based IDCS and the TOA-RSS based IDCS design in terms of MSE performance. This is because incorporating signal-domain information provides a performance boost compared to using only information-domain estimates. Furthermore, even with a limited number of quantization bits for the signal-domain data, the constraint on the CRLB, i.e., $\text{CRLB}_{\hat{\boldsymbol\theta}} \le \epsilon$ , ensures an efficient quantizzation bit allocation scheme that makes the proposed method  outperform the IDCS approach. Besides, it is worth noting that employing a smaller $\epsilon$ in constraint \eqref{eq24b} enhances the MSE performance of the HISDCS design. Especially when $\epsilon$ is set to $1.01\times \epsilon^*$ ($\epsilon^*$ denotes the minimum achievable CRLB as mentioned in Section IV), the proposed HISDCS scheme  outperforms the uniform-quantization SDCS scheme and its performance is remarkably close to the ideal SDCS design, which is regarded as the localization performance lower bound. Besides, from Fig. 4 (b), it is evident that the proposed HISDCS scheme sigificantly reduces the communication overhead as compared to the uniform-quantization SDCS scheme, since the former only needs 8 channel uses, while almost 101 channel uses are required by the latter.

\begin{figure}[t!]
\centerline{\includegraphics[width=1\linewidth]{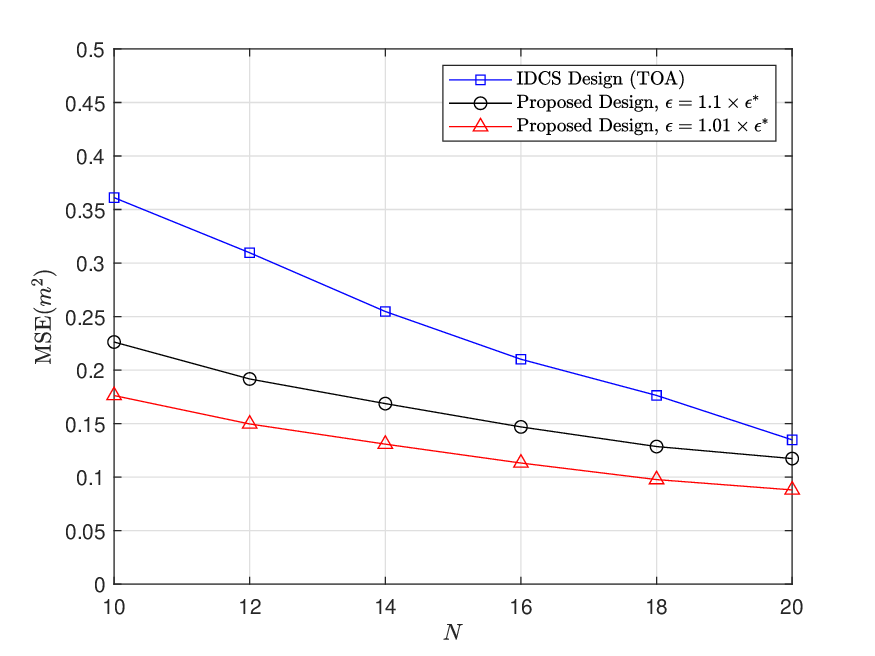}}
\caption{MSE performance comparison under a large number of sensing nodes.}
\label{fig66}
\end{figure}

Fig. \ref{fig_n_scheme} presents the MSE performance and  average channel use number comparison between the baseline designs and the proposed HISDCS design under different numbers of sensing receivers $N$, where we set the average SNR to  0 dB. From Fig. 5 (a), it is seen that the localization performance of the proposed HISDCS scheme is superior to that of the TOA based IDCS and TOA-RSS based schemes. Additionally, as $N$ increases, the MSE of all the considered schemes decreases, which is reasonable since larger $N$ implies higher cooperation gain which is beneficial for localization performance improvement. Besides, from Fig. 5 (b), it can be observed that the communication overhead of the proposed HISDCS scheme is substantially reduced as compared to that required by the uniform-quantization SDCS scheme, and with the increasing of  $N$, there is a consistent decrease in $\bar W$ across all schemes, which is mainly due to the fact that the amount of information transmitted to the FC for localization performance improvement increases as $N$ increases.

Fig. \ref{fig66} shows the MSE performance comparison under a larger number of nodes. As can be seen, the localization performance of the
proposed HISDCS scheme is still superior to that of the TOA based IDCS scheme when $N \ge 10$. This is mainly because the proposed method combines the advantages offered by  information-domain parameters and signal-domain information, and the latter allows for a better extraction of the signal-domain features, and therefore enhances the overall  MSE performance.
However, as the number of nodes increases, the fusion of the delay estimates becomes more accurate, which effectively reduces the impact of individual estimation errors and thereby limits the additional benefits gained from signal-domain processing in the proposed scheme.

\begin{figure}[t!]
\centerline{\includegraphics[width=1\linewidth]{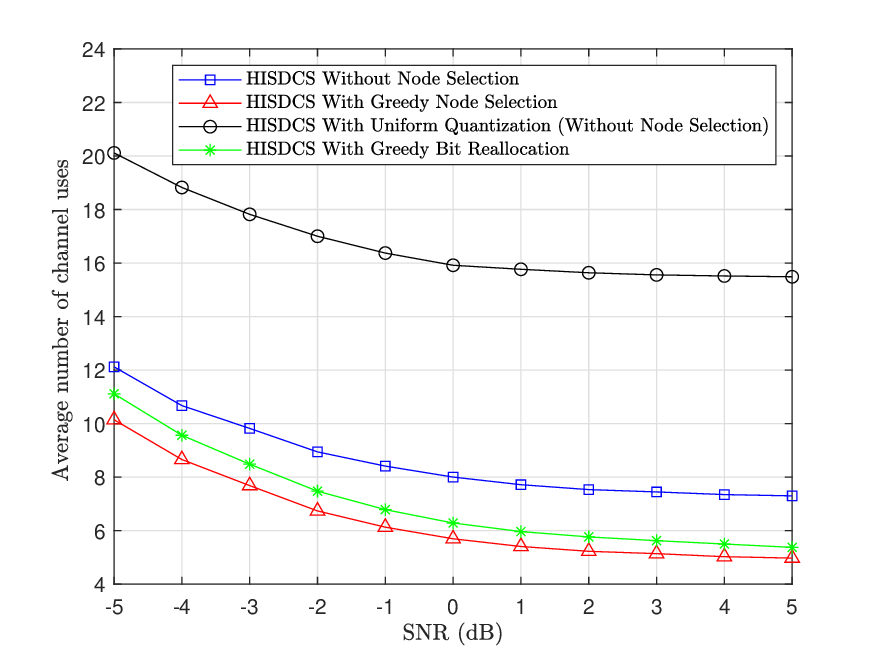}}
\caption{Channel use number comparison under different SNR regimes.}
\label{fig6}
\end{figure}

\begin{figure}[t!]
\centerline{\includegraphics[width=1\linewidth]{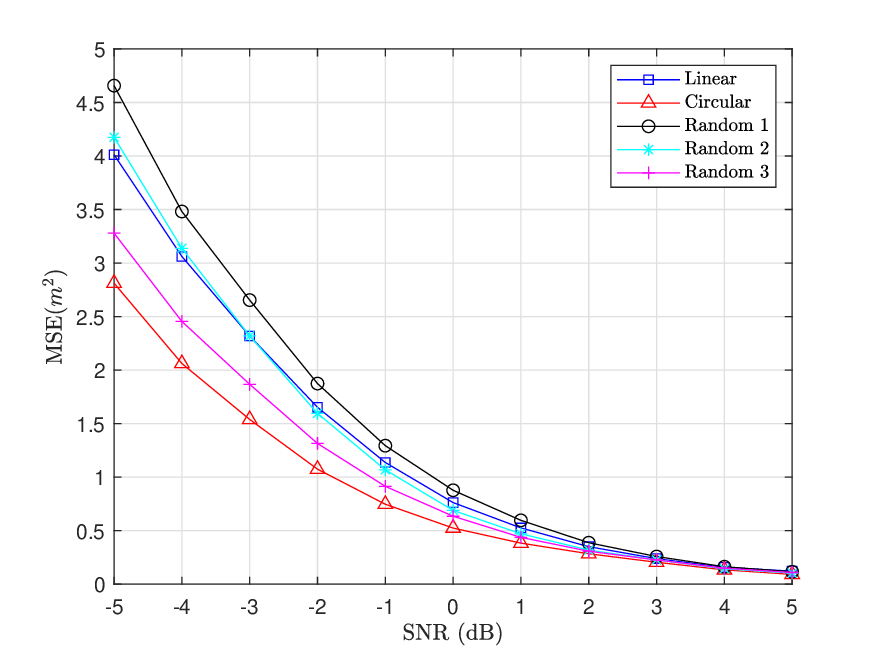}}
\caption{MSE performance comparison under different topologies.}
\label{fig7}
\end{figure}

In Fig. \ref{fig6}, we show the average number of channel uses $\overline{W}$ required by the proposed HISDCS scheme with or without node selection, and compare the performance of the proposed node selection strategy with the low-complexity bit reallocation algorithm presented in Section IV. From this figure, it is seen that compared to the scheme with the same number of quantization bits allocated to each receiver, the proposed MCSCA algorithm is much better in terms of communication cost. Besides, it can be observed that when node selection is employed, the number of channel uses $\overline{W}$ can be substantially reduced, for example, $\overline{W}$ can be reduced by $30.1\%$ when $\textrm{SNR} = 0$ dB, indicating a reduction in communication overhead. Moreover, it  is noteworthy that the proposed greedy bit reallocation algorithm requires a  slightly higher average number of channel use than the scheme with greedy node selection strategy, but it achieves a substantial reduction in computational complexity (as discussed in Section IV).

Next, we present in Fig. \ref{fig7} the MSE performance of the proposed HISDCS scheme  versus  SNR, under different receiver topologies. In the circular topology case, the sensing receivers are symmerically distributed on a 500 m radius circle around the transmitter and the target is uniformly distributed within the circle. We can observe from this figure that the localization performance under the circle topology is better than those under the linear topology and random topologies, which means that the receiver topology should be carefully designed in practice.

\begin{figure}[t!]
\centerline{\includegraphics[width=1\linewidth]{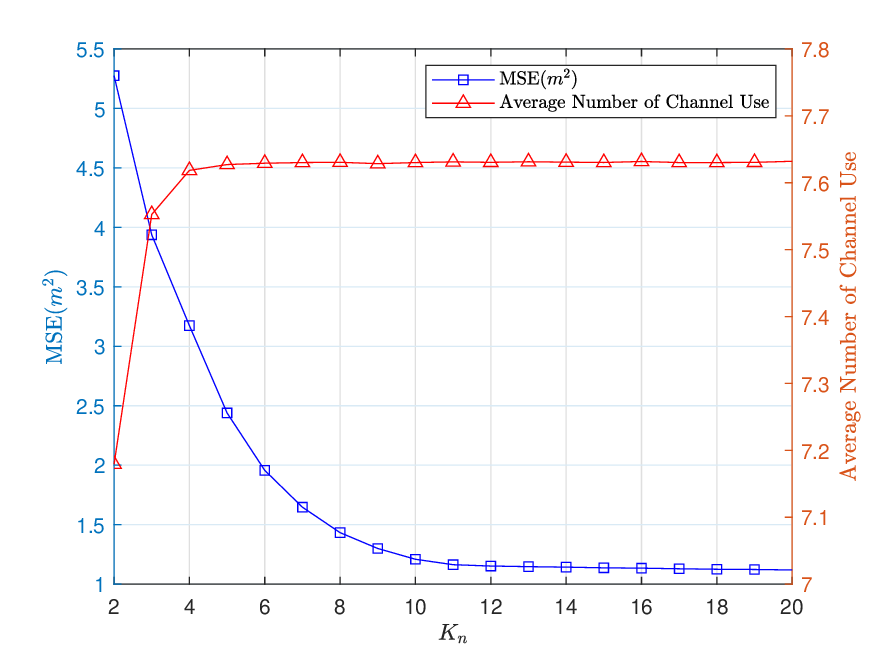}}
\caption{MSE performance and required channel use number versus the number of sampling points $K_n$.}
\label{fig8}
\end{figure}

Then, we show in Fig. \ref{fig8} the mean squared error (MSE) performance for target location estimation and the average channel use number $\overline{W}$ achieved by Algorithm 2 under different numbers of sampling points $K_n$. In this simulation, average signal-noise ratio (SNR) of the receivers is set to be -2 dB. From Fig. \ref{fig8}, it can be seen that the localization performance improves with the increasing of $K_n$. However, when $K_n$ is large, the localization performance improvement gradually saturates, and when $K_n$ exceeds 14, the localization performance reaches a plateau. Besides, we can observe from Fig. 8  that the required number of channel uses $\overline{W}$ first increases as $K_n$ increases, and then gradually stabilizes as $K_n$ surpasses 4. Hence, the value of $K_n$ needs to be carefully designed, for example, setting it to 10, which enables the achievement of favorable localization performance without incurring significant communication overhead, and the complexity of the MCSCA will not be excessively high.

\begin{figure}[t!]
\centerline{\includegraphics[width=1\linewidth]{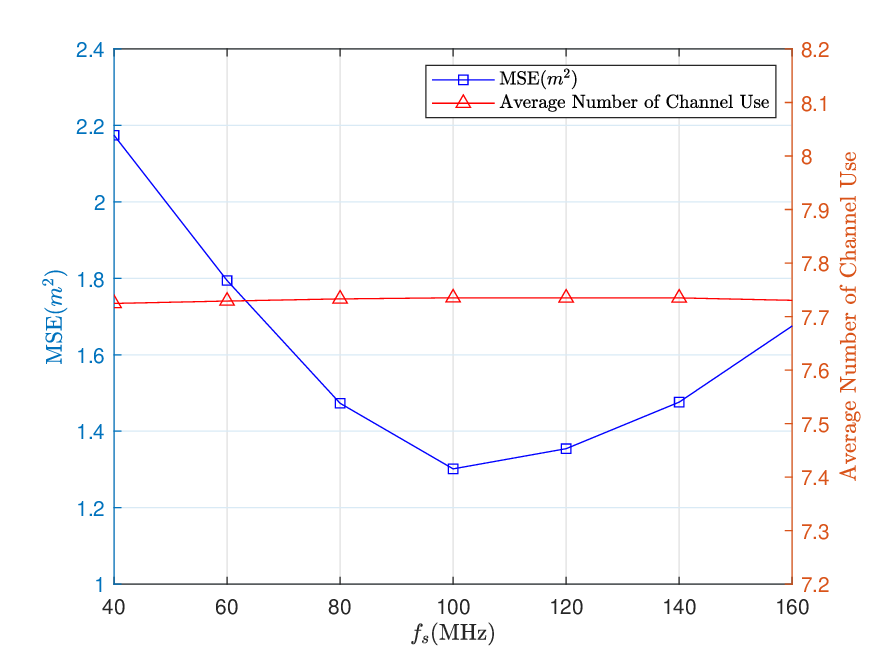}}
\caption{MSE performance and required channel use number versus the sampling frequency $f_s$.}
\label{fig_fs}
\end{figure}

Fig. \ref{fig_fs} illustrates the impact of the sampling frequency $f_s$ on the MSE performance and the average channel use number $\overline{W}$, under the given sampling point number $K_n=10$. It is observed that the best  localization performance is achieved when $f_s$ is about 100 MHz, otherwise, the performance gradually deteriorates. This is mainly due to the fact that, an excessively small or large $f_s$ may result in the inability to extract useful localization information within the main lobe region of the echo signal under given $K_n$, consequently leading to a decline in the localization performance. Besides, we can see that the channel use number $\overline{W}$ barely changes as $f_s$ varies, which is expected since $K_n$ is fixed in this simulation.

\begin{figure}[t!]
\centerline{\includegraphics[width=1\linewidth]{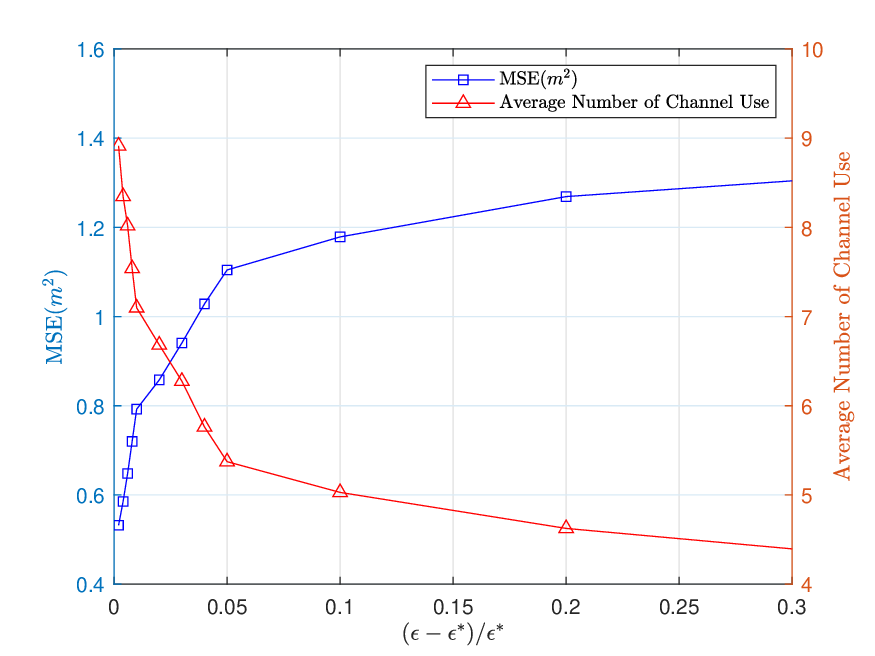}}
\caption{MSE performance and the number of channel use comparison under different values of $\epsilon$.}
\label{fig10}
\end{figure}

Finally, we present in Fig. \ref{fig10} the MSE performance and the average number of channel use $\overline{W}$ achieved by the proposed HISDCS scheme under different values of $\epsilon$. It can be observed that as $\epsilon$ increases, the MSE  gradually increases  while $\overline{W}$ decreases. This is mainly owing to the fact that the  demand for meeting the CRLB constraint gradually diminishes as $\epsilon$ increases, thus leading to localization performance degradation and communication overhead reduction. Hence, there is a trade-off between these two performance metrics that should be carefully considered based on the practical requirements.

\section{Conclusion, Challenges and Future Work}

In this work, we investigated a cooperative sensing optimization problem in a multi-functional network. In order to improve the sensing accuracy while reduce the communication cost, we designed a HISDCS scheme, where each sensing receiver transmits the estimated time delay and effective reflecting coefficient, as well as the received sensing signal sampled around the estimated delay, to the FC. An optimization problem was formulated to minimize the number of channel uses under both CRLB and limited MAC capacity constraints. To tackle this problem, we proposed a MCSCA algorithm for quantization bit allocation and a greedy strategy for node selection. The convergence of the MCSCA algorithm to the set of KKT solutions was theoretically proved. Besides, in order to further reduce the computational complexity, a greedy bit reallocation algorithm was proposed. 
Numerical simulations showed that our proposed HISDCS scheme is able to outperform the IDCS and SDCS schemes.

While the proposed approach has shown promising results, we list several challenges in practical implementation as follows:
\begin{enumerate}
\item{\textbf{Multipath interference:} In practical cooperative sensing and wireless communication systems, multipath fading and interference may degrade the overall performance. Thus, robust signal processing techniques are usually required to address this problem, which may however increase the implementation complexity.}
\item{\textbf{Synchronization:} Achieving accurate synchronization between transmitters and receivers, especially in dynamic environments, can be challenging. This may impact both the precision of signal decoding and the overall system performance.}
\item{\textbf{Hardware overhead:} The proposed solution, particularly with multiple receivers and high-precision quantization, may require a large amount of hardware resources, including logic processing unit, memory, and power consumption, etc. This hardware cost is an important factor that should be considered in real-world applications. }
\end{enumerate}

Despite these challenges, there are several promising directions for future research:
\begin{enumerate}
\item{\textbf{Multi-target detection:} A natural extension of our current work is to further and explore the multi-target detection scenario, where  developing more sophisticated algorithms is required to handle the complexity of detecting multiple targets simultaneously, and this is a common challenge in many real-world applications.}
\item{\textbf{Incorporation of MIMO systems:} Another potential improvement is to integrate the multiple-input multiple-output (MIMO) technology into the proposed HISDCS framework, which is envisioned to further enhance the system capacity and robustness by exploiting the spatial diversity and multiplexing gain offered by MIMO. Investigating the interaction between MIMO and our proposed framework would be a valuable extension.}
\item{\textbf{Trackling More complex channel models:} Extending the current channel model to account for more realistic channel conditions, such as non-line-of-sight (NLOS) propagation, Doppler effects, and dynamic environments, would provide a more accurate representation of real-world scenarios and thus leads to more robust solutions. }
\item{\textbf{Exploring Multi-Sensor Time Delay Based Truncation:} Signal truncation using multi-sensor time delay estimates is an interesting solution to further improve the localization accuracy. By leveraging the estimated target location from multiple sensors, a more precise delay for signal truncation can be computed, which thereby mitigates the issue caused by the inaccurate single-sensor delay estimates and enhances the overall robustness of cooperative sensing.}
\end{enumerate}

\section*{APPENDIX A}
For the considered Gaussian multiple access channel, \eqref{eq23} can be obtained by referring to  the classical capacity theorem in \cite{cover1999elements}, which is given below.
\begin{theorem}[\textbf{m-User Multiple-Access Channel Capacity}]
Consider a multiple-access channel with $m$ users and let $X_i,i=1,\cdots,m$ denote the transmit signal of the $m$ users, then the received signal $Y$ can be expressed as
\begin{equation}
    Y = \sum_{i=1}^{m} X_m + Z,
\end{equation}
where $Z \sim \mathcal{CN}(0,N_0)$ is the complex Gaussian noise.
The capacity region of the m-user multiple-access channel is the closure of the convex hull of the rate vectors satisfying
\begin{equation}
    R(S) \le I(X(S);Y\mid X(S^c))\quad \textrm{for all } S \in \{1,2,\cdots,m\},
    \label{eq59}
\end{equation}
where $S^c$ denotes the complement of $S$, $R(S)=\sum_{i \in S}R_i$, and $X(S)=\{X_i:i\in S\}$.\\
\end{theorem}
Then, we expand the mutual information in \eqref{eq59} in terms of relative entropy as follows:
\begin{equation}
\begin{aligned}
    & I(X(S);Y \mid X(S^c)) \\
    &\overset{(a)}{=}   h(Y \mid X(S^c)) - h(Y \mid X_1,\cdots,X_m)\\
    &= h(X_1+\cdots+X_m+Z \mid X(S^c))\\
    &\quad-h(X_1+\cdots+X_m+Z\mid X_1,\cdots,X_m)\\
    &= h\left(\sum_{i \in S} X_i +Z\right)-h(Z)\\
    &\overset{(b)}{\le} \log\left[2\pi e\left(\sum_{i \in S} P_ig_i + N_0\right)\right]-\log(2\pi e N_0)\\
    &=\log\left(1+\frac{\sum_{i \in S}P_ig_i}{N_0}\right),
\end{aligned}
\label{eq60}
\end{equation}
where $h(\cdot)$ denotes the differential entropy, (a) is due to $I(A;B\mid C)=h(B\mid C)-h(B \mid A,C)$, and (b) follows from the fact that the entropy is maximized when the random variables involved (here $\sum_{i \in S} X_i +Z$ ) are distributed as Gaussian variables with  variance $\sum_{i \in S} P_ig_i + N_0$, and the fact that $h(Z)=\log(2\pi eN_0)$  for Gaussian noise $Z \sim \mathcal{CN}(0,N_0)$ \cite{cover1999elements}. Hence, based on \eqref{eq59} and \eqref{eq60}, we can obtain \eqref{eq23}, which shows that the sum of the individual rates is bounded above by the rate achieved by a single transmitter sending with a power equal to the sum of the individual powers.

\section*{APPENDIX B\\Proof of Lemma 1}

 According to \cite{rudin1964principles}, if we can prove that when $t \to \infty$, the supremum and infimum of $\left\|\bar{\boldsymbol{x}}^{t}-\boldsymbol{x}^{t}\right\|$ are zeroes, then \eqref{eq41} must hold. The proof can be completed via the following two steps.

1. We first prove that $\liminf _{t \rightarrow \infty}\left\|\bar{\boldsymbol{x}}^{t}-\boldsymbol{x}^{t}\right\|=0$.

Due to the strong convexity of $\bar{f}^t_0(\boldsymbol{x})$, we have
\begin{equation}
\begin{aligned}
    \nabla^{T} \bar{f}_{0}^{t}\left(\boldsymbol{x}^{t}\right) \boldsymbol{d}^{t} &\leq-\eta_1\left\|\boldsymbol{d}^{t}\right\|^{2}+\bar{f}_{0}^{t}\left(\bar{\boldsymbol{x}}^{t}\right)-\bar{f}_{0}^{t}\left(\boldsymbol{x}^{t}\right)\\
    &\leq-\eta_1\left\|\boldsymbol{d}^{t}\right\|^{2},
\end{aligned}
\label{eq48}
\end{equation}
where $\eta_1>0$, $\boldsymbol{d}^t=\bar{\boldsymbol{x}}^t-\boldsymbol{x}^t$ and the last inequality is due to the fact $\bar{f}_{0}^{t}\left(\bar{\boldsymbol{x}}^{t}\right)-\bar{f}_{0}^{t}\left(\boldsymbol{x}^{t}\right) \le 0$ always holds since $\bar{\boldsymbol{x}}^t$ is the optimal solution of the problem \eqref{eq40}, whereas $\boldsymbol{x}^t$ is a feasible solution. Besides, since the gradient of ${f}_0(\boldsymbol{x})$ is Lipschitz continuous, we can obtain
\begin{equation}
    \begin{aligned} 
  f_0(\boldsymbol{x}^{t+1})&\overset{(a)}{\le} f_0(\boldsymbol{x}^t)+\eta_2\beta^t\nabla ^Tf_0(\boldsymbol{x}^t)\boldsymbol{d}^t \\
&+L_0(\eta_2)^2(\beta^t)^2\left \| \boldsymbol{d}^t \right \|^2 \\
&\overset{(b)}{\le} f_0(\boldsymbol{x}^t)-\eta_1\eta_2\beta^t\left \| \boldsymbol{d}^t \right \|^2+O(\beta^t), 
\end{aligned}
\end{equation}
where $L_0>0$, (a) is due to $\boldsymbol{x}^{t+1}=\boldsymbol{x}^t+\eta_2\beta^t\boldsymbol{d}^t$, $\eta_2>0$ (\eqref{eq39} is a convex problem) and the Lipschitz gradient continuity of $f_0(\boldsymbol{x})$, and in (b), we use \eqref{eq48} and the fact that $\|\nabla^T\bar{f}_{0}(\boldsymbol{x}^t)-\nabla^T {f}_{0}(\boldsymbol{x}^t)\|=0$. 

Then, we prove $\liminf _{t \rightarrow \infty}\left\|\bar{\boldsymbol{x}}^{t}-\boldsymbol{x}^{t}\right\|=0$ by contradiction. Assuming that there exists a positive constant $\rho$ such that positive $\liminf _{t \rightarrow \infty}\left\|\bar{\boldsymbol{x}}^{t}-\boldsymbol{x}^{t}\right\|\ge \rho >0$ holds, then we can easily find a sequence of $\boldsymbol{d}^t$ that satisfies $\|\boldsymbol{d}^t\|\ge \rho$ for all $t$. Thus by choosing a sufficiently large $t_0$, there always exists $\bar \eta >0$ such that
\begin{equation}
    f_0(\boldsymbol{x}^{t+1})-f_0(\boldsymbol{x}^{t})\le -\beta^t\bar{\eta}\|\boldsymbol{d}^t\|^2,\forall t \ge t_0,
    \label{eq50}
\end{equation}
therefore, it follows from \eqref{eq50} that 
\begin{equation}
    f_0(\boldsymbol{x}^t)-f_0(\boldsymbol{x}^{t_0})\le -\bar{\eta}(\rho)^2\sum_{j=t_0}^{t} \beta^j.
    \label{eq51}
\end{equation}
By letting $t \to \infty$, \eqref{eq51} contradicts the boundedness of $\{f_0(\boldsymbol{x}^t)\}$ given the fact that $\sum_{j=t_0}^{\infty} \beta^j=\infty$. Therefore, we have $\liminf _{t \rightarrow \infty}\left\|\bar{\boldsymbol{x}}^{t}-\boldsymbol{x}^{t}\right\|=0$.

2. Then, we prove $\limsup _{t \rightarrow \infty}\left\|\bar{\boldsymbol{x}}^{t}-\boldsymbol{x}^{t}\right\|=0$.

It follows from the Lipschitz continuity and strong convexity of $\bar{f}^t_i(\boldsymbol{x}),i=0,\cdots,m$ that\cite{liu2019stochastic} 
\begin{equation}
    \left \| \bar{\boldsymbol{x}}^{t_1} -\bar{\boldsymbol{x}}^{t_2}\right \| \le \hat{L}\left \| {\boldsymbol{x}}^{t_1}-{\boldsymbol{x}}^{t_2} \right \|+e(t_1,t_2) ,
    \label{eq52}
\end{equation} 
where $\hat L>0$ and $\lim_{t_1,t_2 \to \infty} e(t_1,t_2)=0$. Then, similar to the previous step, we prove $\limsup _{t \rightarrow \infty}\left\|\bar{\boldsymbol{x}}^{t}-\boldsymbol{x}^{t}\right\|=0$ by contradiction. 

Suppose that $\limsup _{t \rightarrow \infty}\left\|\bar{\boldsymbol{x}}^{t}-\boldsymbol{x}^{t}\right\|>0$, since we have proved $\liminf _{t \rightarrow \infty}\left\|\bar{\boldsymbol{x}}^{t}-\boldsymbol{x}^{t}\right\|=0$, it follows that there exists a $\delta >0$ such that both $\|\boldsymbol{d}^t\| \ge 2\delta$ and $\|\boldsymbol{d}^t\| < \delta$ holds for infinitely many $t$. Thus, we can always find an infinite set of indexes, denoted by $\mathcal{T}$, which satisfies the following property: for any $\forall t \in \mathcal{T}$, there exists $t_2>t$, such that
\begin{equation}
    \begin{aligned}
         &\left \| \boldsymbol{d}^t \right \| \le \delta,\\
&\left \| \boldsymbol{d}^{t_2} \right \| \ge 2\delta,\\
&\delta <  \left \| \boldsymbol{d}^n \right \| < 2\delta, t<n<t_2,
    \end{aligned}
    \label{eq53}
\end{equation}
holds. Then, based on \eqref{eq53}, we have
\begin{equation}
    \begin{aligned} 
 \delta &\le \left \| \boldsymbol{d}^{t_2} \right \| -\left \| \boldsymbol{d}^{t} \right \|\\
& \le \left \| \boldsymbol{d}^{t_2} -\boldsymbol{d}^{t}\right \| \\
& = \left \| (\bar{\boldsymbol{x}}^{t_2}-\bar{\boldsymbol{x}}^{t})-({\boldsymbol{x}}^{t_2}-{\boldsymbol{x}}^{t}) \right \|\\
& \le  \left \| \bar{\boldsymbol{x}}^{t_2}-\bar{\boldsymbol{x}}^{t}\right\|+\left\|{\boldsymbol{x}}^{t_2}-{\boldsymbol{x}}^{t}\right \|\\
& \overset{(a)}{\le} (1+\hat L)\left\|{\boldsymbol{x}}^{t_2}-{\boldsymbol{x}}^{t}\right \|+e(t_2,t)\\
& = (1+\hat L)\|{\boldsymbol{x}}^{t_2}-{\boldsymbol{x}}^{t_2-1}+{\boldsymbol{x}}^{t_2-1}-{\boldsymbol{x}}^{t_2-2}+\cdots\\
&+{\boldsymbol{x}}^{t+1}-{\boldsymbol{x}}^{t} \|+e(t_2,t)\\
& \le (1+\hat L)\sum_{n=t}^{t_2-1} \hat{\eta}\beta^n\|\boldsymbol{d}^n\| +e(t_2,t)\\
& \le 2\delta \hat{\eta}(1+\hat L)\sum_{n=t}^{t_2-1} \beta^n+e(t_2,t),
\end{aligned}
\end{equation}
where $\hat \eta>0$ and  (a) is obtained by resorting to \eqref{eq52}. Then we obtain
\begin{equation}
    \liminf_{t\to \infty} \sum_{n=t}^{t_2-1} \beta^n \ge \delta_1 = \frac{1}{2\hat\eta(1+\hat L)} >0.
    \label{eq55}
\end{equation}
Invoking \eqref{eq50}, there exists $\bar \eta>0$, for $\forall t \in \mathcal{T}$, such that
\begin{equation}
    f_0(\boldsymbol{x}^{n+1})-f_0(\boldsymbol{x}^{n})\le -\beta^n\bar{\eta}\|\boldsymbol{d}^n\|^2,t<n<t_2,
\end{equation}
then for sufficiently large $t$ we have
\begin{equation}
    f_0(\boldsymbol{x}^{t_2})-f_0(\boldsymbol{x}^{t})\le -\bar{\eta}\delta^2\sum_{n=t}^{t_2-1} \beta^n.
\end{equation}
Since the convergence of $\{f_0(\boldsymbol{x}^t)\}$, there must be $\liminf_{t\to \infty} \sum_{n=t}^{t_2-1} \beta^n = 0$ which contradicts \eqref{eq55}. Therefore, it can be shown that $\limsup _{t \rightarrow \infty}\left\|\bar{\boldsymbol{x}}^{t}-\boldsymbol{x}^{t}\right\|=0$.

Combining the results of the above two steps, we have $\lim _{t \rightarrow \infty}\left\|\bar{\boldsymbol{x}}^{t}-\boldsymbol{x}^{t}\right\|=0$, which completes the proof.

\bibliographystyle{IEEEtran}
\bibliography{myref.bib}

\vfill

\end{document}